\documentclass[journal]{IEEEtran}

\usepackage{amsthm}
\usepackage{amsmath,amssymb,amsfonts}
\usepackage{bm}
\usepackage{mathrsfs}
\usepackage{graphicx}
\usepackage{algorithm}
\usepackage{algpseudocode}
\usepackage[caption=false,font=footnotesize]{subfig}
\usepackage{microtype}
\usepackage{booktabs}
\usepackage{cite}
\usepackage{url}
\usepackage[hidelinks]{hyperref}
\usepackage{orcidlink}

\newtheorem{lemma}{Lemma}
\newtheorem{proposition}{Proposition}
\newtheorem{theorem}{Theorem}
\newtheorem{remark}{Remark}
\newtheorem{assumption}{Assumption}
\newtheorem{corollary}{Corollary}
\newtheorem{definition}{Definition}
\newtheorem{problem}{Problem}

\title{Anisotropic Template Ans\"{a}tze for Robust Positive\\
Invariance under State-Dependent Uncertainty}

\author{Abdelrahman Ramadan,~\IEEEmembership{Graduate Student Member,~IEEE,}
        Melissa Greeff,~\IEEEmembership{Member,~IEEE,}\\
        and Sidney Givigi,~\IEEEmembership{Senior Member,~IEEE}
\thanks{A. Ramadan and M. Greeff are with Electrical and Computer Engineering, Smith Engineering, and with Ingenuity Labs Research Institute, Queen's University, Kingston, ON K7L 3N6, Canada (e-mail: 20amr3@queensu.ca; melissa.greeff@queensu.ca).}
\thanks{S. Givigi is with the School of Computing and with Ingenuity Labs Research Institute, Queen's University, Kingston, ON K7L 3N6, Canada (e-mail: sidney.givigi@queensu.ca).}
\thanks{This is an extended version of an article accepted for publication in the IEEE Control Systems Letters, available at \url{https://ieeexplore.ieee.org/document/11557324}. It adds full proofs in the appendix, a complexity analysis, and additional simulation studies. An interactive companion is at \url{https://thelastpixie.github.io/ansatz-rpi-v2/}, with the quadrotor simulator at \url{https://thelastpixie.github.io/ansatz-rpi-v2/sim.html}.}}

\begin{document}

\maketitle

\begin{abstract}
We establish sufficient conditions for robust positive invariance under state- and input-dependent disturbances with anisotropic covariance structure. The proposed ansatz maps a fixed ellipsoidal template through a GP-derived positive-definite matrix field, subsuming scalar homothetic scaling while retaining finite graph-based verification. The resulting LMI conditions couple the learned field to Schur-stable dynamics; an isotropic fallback with inflation factor $r=1/(1-\gamma_{\mathrm{cl}})$ proves admissibility. During each learning epoch the field is frozen, so online tube evaluation is one GP covariance query and a small matrix square root, with no online set iteration or LMI solve. Quadrotor simulations show a $195\times$ reduction in 3D velocity-tube volume and a $2.1{\times}10^5$ reduction in the joint 7D velocity-control subspace relative to a non-adaptive homothetic baseline. This extended version adds full proofs, a separated offline/online complexity analysis, and controller-sweep, contraction, and projection-area studies.
\end{abstract}


\section{Introduction}\label{sec:intro}
\IEEEPARstart{R}{obust} linear Model Predictive Control (MPC) requires disturbance-invariant sets for constraint satisfaction under uncertainty. The standard formulation assumes a fixed, state-independent disturbance set $\mathbb{W}\subset\mathbb{R}^n$ and computes a Robust Positively Invariant (RPI) set $\boldsymbol{\Omega}$ satisfying $\mathbf{A}_{\mathrm{cl}}\boldsymbol{\Omega} \oplus \mathbb{W} \subseteq \boldsymbol{\Omega}$~\cite{kolmanovsky1998theory,blanchini2008set}. This is computationally tractable but conservative: a single $\mathbb{W}$ must bound worst-case disturbances uniformly over the state-input domain, disregarding spatial variation in uncertainty magnitude and orientation.

We observe that methods addressing this conservatism share a common structure: rather than computing a new set at every operating point, each starts from a fixed reference shape and \emph{transforms} it (scaling, stretching, or rotating) to match local uncertainty. We formalize this as a \emph{template ans\"{a}tze}, a structured trial form wherein $(\mathbf{x},\mathbf{u})$-dependence is encoded in transformation parameters rather than set geometry. Existing instantiations differ along three axes: \textbf{(i)} transformation class (scalar, diagonal, full matrix), \textbf{(ii)} dependence structure (prediction stage $k$ versus operating point $(\mathbf{x},\mathbf{u})$), and \textbf{(iii)} parameter source (analytical or data-driven). Homothetic tube MPC~\cite{rakovicHomotheticTubeModel2012, 10327735} uses $\alpha_k \bar{\boldsymbol{\Omega}}$ with scalar $\alpha_k$ optimized over the prediction horizon. The work~\cite{kohlerComputationallyEfficientRobust2021} extends to $(\mathbf{x},\mathbf{u})$-dependent scalar bounds via Lipschitz analysis, while~\cite{gaoLearningbasedHomotheticTube2025} learns stage-indexed $\alpha_{i|k}$ through the scenario approach; these scalar methods capture magnitude variation but not directional structure. Elastic tubes~\cite{rakovicElasticTubeModel2016} generalize to per-facet scaling $\boldsymbol{\gamma}_k\in\mathbb{R}^p$, and configuration-constrained tubes~\cite{villanuevaConfigurationConstrainedTubeMPC2024} optimize polytope vertices online. These permit richer geometric variation but remain stage-indexed, so adaptation tracks the nominal trajectory rather than the current operating point. GP-based methods~\cite{koller2018learning,hewing2020cautious} construct confidence regions from posteriors. The formulation in~\cite{solopertoLearningBasedRobustModel2018} produces hyperboxes $\mathbb{W}(\mathbf{x}) = \{\mathbf{w} \in \mathbb{R}^n : |w_j| \leq \beta\sigma_{n,j}(\mathbf{x}),\; j=1,\ldots,n\}$, equivalent to $\mathbf{P} = \mathrm{diag}(\sigma_1,\ldots,\sigma_n)$, but cannot represent inter-component correlations. The missing case is a learned full matrix $\mathbf{P}(\mathbf{x},\mathbf{u})\in\mathbb{S}^n_{++}$ that varies with the current operating point and admits finite invariance verification. Table~\ref{tab:comparison} summarizes the comparison.

\begin{table}[!t]
\centering
\caption{Template Ans\"{a}tze in the Literature and the Proposed Generalization}
\label{tab:comparison}
\footnotesize
\begin{tabular}{lccc}
\toprule
\textbf{Method} & \textbf{Transformation} & \textbf{Depends on} & \textbf{Learned} \\
\midrule
Homothetic~\cite{rakovicHomotheticTubeModel2012, 10327735} & Scalar $\alpha$ & Stage $k$ & No \\
Elastic~\cite{rakovicElasticTubeModel2016} & Per-facet $\boldsymbol{\gamma}$ & Stage $k$ & No \\
CCTMPC~\cite{villanuevaConfigurationConstrainedTubeMPC2024} & Full shape & Stage $k$ & No \\
K\"ohler et al.~\cite{kohlerComputationallyEfficientRobust2021,sasfiRobustAdaptiveMPC2023} & Scalar & $(\mathbf{x},\mathbf{u})$ & No, Yes \\
GP-MPC~\cite{solopertoLearningBasedRobustModel2018} & Diagonal & $\mathbf{x}$ only & Yes \\
Scenario~\cite{gaoLearningbasedHomotheticTube2025} & Scalar & Stage $k$ & Yes \\
\textbf{Proposed} & \textbf{Full matrix $\mathbf{P}$} & $(\mathbf{x},\mathbf{u})$ & \textbf{Yes} \\
\bottomrule
\end{tabular}
\end{table}

State- and input-dependent disturbances induce a circular dependence: verifying that states remain within a candidate invariant set requires knowledge of reachable states, which depends on the invariant set. Our prior work~\cite{ramadanLearningBasedShrinkingDisturbanceInvariant2025a} resolves this via lifting to an augmented space $\mathbb{R}^{2n+m}$ where the disturbance law becomes a graph constraint, followed by RPI set computation as fixed points of a set-valued operator. This lifted method is exact but computationally expensive, requiring seconds (GPU) to minutes per polytope iteration, precluding real-time use. The present work takes a different approach: the base template pair $(\bar{\mathbf{W}}, \bar{\boldsymbol{\Omega}})$ is computed once analytically at initialization; thereafter, invariance is enforced by constraining the parameter field $\mathbf{P}(\cdot)$ rather than iterating on sets. The same template ans\"{a}tze viewpoint also suggests a hierarchy of increasing expressiveness: scalar (homothetic), diagonal, symmetric positive-definite (present work), general linear, and polynomial.

\paragraph{Contributions} We introduce the \emph{anisotropic template ans\"{a}tze}
\begin{equation}
\mathbb{W}(\mathbf{x},\mathbf{u}) = \mathbf{P}(\mathbf{x},\mathbf{u})\,\bar{\mathbf{W}},
\label{eq:aniso_ansatz}
\end{equation}
where $\bar{\mathbf{W}} \subset \mathbb{R}^n$ is a fixed template and $\mathbf{P}(\mathbf{x},\mathbf{u})\in\mathbb{S}^n_{++}$ encodes local anisotropy learned from GP posteriors via $\mathbf{P} = c_{n,\alpha}\hat{\boldsymbol{\Sigma}}_w^{1/2}$. The contributions are:
\textbf{(C1)}~\textbf{Template ans\"{a}tze formalism:} unification of scalar, diagonal, and full-matrix parameterizations under a single framework indexed by transformation class.
\textbf{(C2)}~\textbf{Parameter-space invariance conditions:} sufficient conditions for RPI as LMIs coupling $\mathbf{P}(\cdot)$ to closed-loop dynamics, bypassing set-valued iteration.
\textbf{(C3)}~\textbf{Graph-based finite verification:} discretization of the operating space into a finite graph with arc-wise conditions that conservatively imply invariance.
\textbf{(C4)}~\textbf{Single-query online evaluation:} after offline GP training and graph verification, each online tube cross-section is obtained by one GP covariance query and a small matrix square root.

\paragraph{Notation} $\mathbb{S}^n_{++}$ and $\mathbb{S}^n_+$ denote symmetric positive-definite and positive-semidefinite matrices; $\mathcal{B}_2^n$ is the Euclidean unit ball; $\oplus$ is the Minkowski sum; $\ominus$ is the Pontryagin difference; $\preceq$ is the L\"owner order; and $\gamma_{\mathbf{S}}$ denotes the gauge function of a set $\mathbf{S}$.

\section{Anisotropic Template Geometry}
\label{sec:prelim}

\subsubsection{System Model and Error Dynamics}

Consider a discrete-time system with nominal linear dynamics and additive state- and input-dependent uncertainty
\begin{equation}
\mathbf{x}_{k+1} = \mathbf{A}\mathbf{x}_k + \mathbf{B}\mathbf{u}_k + \mathbf{w}(\mathbf{x}_k, \mathbf{u}_k),
\label{eq:dynamics}
\end{equation}
where $\mathbf{x}_k \in \mathcal{X}$, $\mathbf{u}_k \in \mathcal{U}$, $\mathbf{A},\mathbf{B}$ are nominal matrices, and $\mathbf{w}(\mathbf{x},\mathbf{u}) \in \mathbb{W}(\mathbf{x},\mathbf{u})$ is a state- and input-dependent residual. Thus $\mathbf{A},\mathbf{B}$ are assumed available from modeling, local linearization, or identification; in the latter case $\mathbf{w}$ also absorbs identification error.

Under the affine feedback $\mathbf{u} = \mathbf{K}\mathbf{x} + \mathbf{v}$ with feedforward $\mathbf{v} \in \mathbb{V} \subseteq \mathbb{R}^m$, the error $\mathbf{e}_k := \mathbf{x}_k - \hat{\mathbf{x}}_k$ relative to a nominal trajectory $\hat{\mathbf{x}}_{k+1} = \mathbf{A}_{\mathrm{cl}}\hat{\mathbf{x}}_k + \mathbf{B}\mathbf{v}_k$ ($\mathbf{A}_{\mathrm{cl}} := \mathbf{A}+\mathbf{B}\mathbf{K}$) satisfies
\begin{equation}
\mathbf{e}_{k+1} = \mathbf{A}_{\mathrm{cl}}\mathbf{e}_k + \mathbf{w}(\mathbf{x}_k, \mathbf{u}_k),
\label{eq:error_dynamics}
\end{equation}
the standard tube MPC error equation with $(\mathbf{x},\mathbf{u})$-dependent disturbances.

\subsubsection{Anisotropic Template Ans\"{a}tze}

Classical homothetic methods use $\mathbb{W}_k = \alpha_k \bar{\mathbf{W}}$ with scalar $\alpha_k \in \mathbb{R}_{>0}$ varying along the prediction horizon ($k=0,\ldots,N{-}1$), capturing magnitude growth but not directional variation. We generalize to full matrix transformations.

\begin{definition}[Anisotropic Template Ans\"{a}tze]
\label{def:aniso_template}
An \emph{anisotropic template ans\"{a}tze} is a family $\mathbb{W}(\mathbf{x},\mathbf{u}) = \mathbf{P}(\mathbf{x},\mathbf{u})\bar{\mathbf{W}} := \{\mathbf{P}(\mathbf{x},\mathbf{u})\mathbf{w} : \mathbf{w} \in \bar{\mathbf{W}}\}$ where $\bar{\mathbf{W}} \subset \mathbb{R}^n$ is a fixed template with $\mathbf{0} \in \mathrm{int}(\bar{\mathbf{W}})$ and $\mathbf{P}: \mathcal{X} \times \mathcal{U} \to \mathbb{S}^n_{++}$ maps to the cone of symmetric positive-definite matrices.
\end{definition}

The matrix $\mathbf{P}(\mathbf{x},\mathbf{u})$ encodes local anisotropy via its spectral decomposition: eigenvectors define principal uncertainty directions; eigenvalues set magnitudes along these axes. Setting $\mathbf{P} = \sigma \mathbf{I}_n$ with $\sigma > 0$ recovers isotropic scaling~\cite{rakovicHomotheticTubeModel2012}; $\mathbf{P} = \mathrm{diag}(\sigma_1,\ldots,\sigma_n)$ yields axis-aligned scaling~\cite{solopertoLearningBasedRobustModel2018}; the general case admits rotation and shear.

From trajectory data, disturbance samples $\mathbf{w}^{(j)} = \mathbf{x}^{(j+1)} - \mathbf{A}\mathbf{x}^{(j)} - \mathbf{B}\mathbf{u}^{(j)}$ yield dataset $\mathcal{D} = \{(\mathbf{z}^{(j)}, \mathbf{w}^{(j)})\}_{j=1}^N$ with $\mathbf{z} := (\mathbf{x},\mathbf{u}) \in \mathcal{Z} := \mathcal{X} \times \mathcal{U}$. Under a multivariate GP prior with posterior covariance $\hat{\boldsymbol{\Sigma}}_w(\mathbf{x},\mathbf{u}) \in \mathbb{S}^n_{+}$, the $(1{-}\alpha)$ credible region is an ellipsoid admitting the decomposition (see Section~\ref{sec:learning})
\begin{equation}
\mathcal{E}(\mathbf{x},\mathbf{u}) = \hat{\boldsymbol{\mu}}_w \oplus c_{n,\alpha}\hat{\boldsymbol{\Sigma}}_w^{1/2}\mathcal{B}_2^n,
\label{eq:credible_ellipsoid}
\end{equation}
where $c_{n,\alpha} := \sqrt{\chi^2_{n,1-\alpha}}$ and $\mathcal{B}_2^n$ is the Euclidean unit ball. Following standard practice in GP-based MPC, the posterior mean $\hat{\boldsymbol{\mu}}_w$ is absorbed into a corrected nominal model, leaving zero-mean residual uncertainty. The GP-induced anisotropy field is then $\mathbf{P}(\mathbf{z}) := c_{n,\alpha}\hat{\boldsymbol{\Sigma}}_w(\mathbf{z})^{1/2}$. Writing the spectral decomposition $\hat{\boldsymbol{\Sigma}}_w = \mathbf{V}\boldsymbol{\Lambda}\mathbf{V}^\top$, we obtain $\mathbf{P} = c_{n,\alpha}\mathbf{V}\boldsymbol{\Lambda}^{1/2}\mathbf{V}^\top$: as $(\mathbf{x},\mathbf{u})$ varies, both orientation ($\mathbf{V}$) and shape ($\boldsymbol{\Lambda}^{1/2}$) change, a structure that scalar and diagonal parameterizations cannot represent.

\subsubsection{Base Template and Invariance Geometry}

For state- and input-dependent disturbances, the computation of RPI sets is inherently recursive: reachable states depend on disturbance sets, which depend on reachable states. We fix $\bar{\mathbf{W}} = \mathcal{B}_2^n$ and define the \emph{closed-loop contraction rate} $\gamma_{\mathrm{cl}} < 1$: if $\|\mathbf{A}_{\mathrm{cl}}\|_2 < 1$ set $\gamma_{\mathrm{cl}} := \|\mathbf{A}_{\mathrm{cl}}\|_2$; otherwise set $\gamma_{\mathrm{cl}} := \sqrt{\rho(\mathbf{P}_{\mathrm{d}}^{-1}\mathbf{A}_{\mathrm{cl}}^\top\mathbf{P}_{\mathrm{d}}\,\mathbf{A}_{\mathrm{cl}})}$ where $\mathbf{P}_{\mathrm{d}} \succ \mathbf{0}$ is the DARE terminal-cost matrix. Schur stability guarantees $\gamma_{\mathrm{cl}} < 1$ in both cases. The \emph{template inflation factor}
\begin{equation}
r \;:=\; \frac{1}{1 - \gamma_{\mathrm{cl}}}
\label{eq:template_inflation}
\end{equation}
is minimal subject to $\gamma_{\mathrm{cl}}\,r + 1 \leq r$, ensuring $\bar{\boldsymbol{\Omega}} := r\mathcal{B}_2^n$ is RPI for $(\mathbf{A}_{\mathrm{cl}}, \mathcal{B}_2^n)$. When $\|\mathbf{A}_{\mathrm{cl}}\|_2 \geq 1$, the coordinate change $\tilde{\mathbf{e}} = \mathbf{P}_{\mathrm{d}}^{1/2}\mathbf{e}$ yields $\|\tilde{\mathbf{A}}_{\mathrm{cl}}\|_2 = \gamma_{\mathrm{cl}} < 1$; all subsequent spectral bounds $p_i^{\pm}$ and template norms are understood in the $\mathbf{P}_{\mathrm{d}}$-weighted inner product. This base template pair $(\bar{\mathbf{W}}, \bar{\boldsymbol{\Omega}}) = (\mathcal{B}_2^n, r\mathcal{B}_2^n)$ is computed once at initialization; subsequent analysis constrains only $\mathbf{P}(\cdot)$.

\begin{assumption}[Anisotropy Field Regularity]
\label{asm:P_regularity}
The map $\mathbf{P}: \mathcal{Z} \to \mathbb{S}^n_{++}$ satisfies uniform L\"owner bounds $p_{\min}\mathbf{I}_n \preceq \mathbf{P}(\mathbf{z}) \preceq p_{\max}\mathbf{I}_n$ with $0 < p_{\min} \leq p_{\max} < \infty$, and is Lipschitz continuous with constant $L_P > 0$.
\end{assumption}

For a convex compact set $\mathbf{S} \subset \mathbb{R}^n$ with $\mathbf{0} \in \mathrm{int}(\mathbf{S})$, the \emph{gauge function} (Minkowski functional) is $\gamma_{\mathbf{S}}(\mathbf{w}) := \inf\{\lambda \geq 0 : \mathbf{w} \in \lambda \mathbf{S}\}$, satisfying $\gamma_{\mathbf{S}}(\mathbf{w}) \leq 1 \Leftrightarrow \mathbf{w} \in \mathbf{S}$. A set $\bar{\boldsymbol{\Omega}}$ is \emph{template RPI} for $(\mathbf{A}_{\mathrm{cl}}, \bar{\mathbf{W}})$ if $\mathbf{A}_{\mathrm{cl}}\bar{\boldsymbol{\Omega}} \oplus \bar{\mathbf{W}} \subseteq \bar{\boldsymbol{\Omega}}$; the minimal such set exists when $\mathbf{A}_{\mathrm{cl}}$ is Schur stable~\cite{kolmanovsky1998theory}.

\begin{definition}[Template Contraction Coefficient]
\label{def:contraction}
For Schur-stable $\mathbf{A}_{\mathrm{cl}}$ and template $\bar{\boldsymbol{\Omega}} \ni \mathbf{0}$, the contraction coefficient is $\rho_{\bar{\boldsymbol{\Omega}}} := \sup_{\mathbf{e} \in \bar{\boldsymbol{\Omega}}} \gamma_{\bar{\boldsymbol{\Omega}}}(\mathbf{A}_{\mathrm{cl}} \mathbf{e})$.
\end{definition}

\begin{lemma}[Strict Contraction]
\label{lem:contraction}
If $\bar{\boldsymbol{\Omega}}$ is RPI for $(\mathbf{A}_{\mathrm{cl}}, \bar{\mathbf{W}})$ with $\mathbf{0} \in \mathrm{int}(\bar{\mathbf{W}})$, then $\rho_{\bar{\boldsymbol{\Omega}}} < 1$.
\end{lemma}

The proof is in Appendix~\ref{app:contraction}.

\begin{definition}[Local RPI Field]
\label{def:local_rpi_field}
A map $\mathbf{x} \mapsto \mathcal{B}(\mathbf{x}) \subset \mathbb{R}^n$ (not to be confused with $\mathcal{B}_2^n$) assigning a compact set to each state is a \emph{local RPI field} if the graph $\mathcal{O} := \bigcup_{\mathbf{x} \in \mathcal{X}}(\mathbf{x} + \mathcal{B}(\mathbf{x}))$ is RPI for the closed-loop system: $\mathbf{x}_k \in \mathcal{O} \Rightarrow \mathbf{x}_{k+1} \in \mathcal{O}$ for all admissible $\mathbf{w}_k \in \mathbb{W}(\mathbf{x}_k, \mathbf{u}_k)$.
\end{definition}

Under the template ans\"{a}tze, we seek $\mathcal{B}(\mathbf{x}) = \mathbf{P}(\mathbf{x},\mathbf{K}\mathbf{x})\bar{\boldsymbol{\Omega}}$: scaled, rotated copies of the base template.

\begin{problem}[Anisotropic Template Invariance]
\label{prob:main}
Given dynamics~\eqref{eq:dynamics}, template pair $(\bar{\mathbf{W}}, \bar{\boldsymbol{\Omega}}) = (\mathcal{B}_2^n, r\mathcal{B}_2^n)$ with $r$ as in~\eqref{eq:template_inflation}, and GP-induced field $\mathbf{P}(\mathbf{z})$ satisfying Assumption~\ref{asm:P_regularity}, find conditions on $\mathbf{P}(\cdot)$ such that $\boldsymbol{\Omega}(\mathbf{z}) := \mathbf{P}(\mathbf{z})\bar{\boldsymbol{\Omega}}$ is a local RPI field.
\end{problem}

Unlike~\cite{ramadanLearningBasedShrinkingDisturbanceInvariant2025a}, we do not iterate on sets. Section~\ref{sec:invariance} derives constraints directly on $\mathbf{P}(\cdot)$, a shift from set-valued to function-valued fixed points enabling single-query online evaluation.

\section{Learning the Anisotropy Field from Data}
\label{sec:learning}

This section constructs $\mathbf{P}(\mathbf{z})$ from trajectory data via Gaussian Process regression and establishes the regularity properties required for invariance.

\subsubsection{GP Disturbance Model}

From observed transitions $\{(\mathbf{x}^{(j)}, \mathbf{u}^{(j)}, \mathbf{x}^{(j+1)})\}_{j=1}^{N}$, disturbance samples $\mathbf{w}^{(j)} := \mathbf{x}^{(j+1)} - \mathbf{A}\mathbf{x}^{(j)} - \mathbf{B}\mathbf{u}^{(j)}$ capture model mismatch at $\mathbf{z}^{(j)} := (\mathbf{x}^{(j)}, \mathbf{u}^{(j)})$, yielding $\mathcal{D} = \{(\mathbf{z}^{(j)}, \mathbf{w}^{(j)})\}_{j=1}^{N}$. The GP therefore learns residual uncertainty around a nominal model, not a fully model-free controller. The disturbance is modeled via a multi-output GP with kernel $k:\mathcal{Z}\times\mathcal{Z}\to\mathbb{R}$.

\begin{assumption}[Kernel Regularity]
\label{asm:kernel}
The kernel $k$ is bounded ($k(\mathbf{z},\mathbf{z}) \leq \bar{k}$ for all $\mathbf{z} \in \mathcal{Z}$) and Lipschitz ($|k(\mathbf{z}_1,\mathbf{z}') - k(\mathbf{z}_2,\mathbf{z}')| \leq L_k\|\mathbf{z}_1 - \mathbf{z}_2\|$ for all $\mathbf{z}_1, \mathbf{z}_2, \mathbf{z}' \in \mathcal{Z}$). The active GP dictionary is finite, $N_{\mathrm{gp}}<\infty$, the observation noise satisfies $\sigma_n^2>0$, and the LMC prior covariance is uniformly nondegenerate on $\mathcal{Z}$. The boundedness and Lipschitz parts hold for the squared-exponential kernel $k_{\mathrm{SE}}(\mathbf{z},\mathbf{z}') = \sigma_f^2\exp(-\|\mathbf{z}-\mathbf{z}'\|^2/2\ell^2)$, with $L_k = \sigma_f^2/\ell$.
\end{assumption}

Conditioning on $\mathcal{D}$ with observation noise variance $\sigma_n^2 > 0$ yields posterior mean $\hat{\boldsymbol{\mu}}_w(\mathbf{z})$ and covariance $\hat{\boldsymbol{\Sigma}}_w(\mathbf{z}) \in \mathbb{S}^n_{+}$. For independent output components, the posterior variance is $\hat{\sigma}_{w,i}^2(\mathbf{z}) = k(\mathbf{z},\mathbf{z}) - \mathbf{k}_*^\top(\mathbf{G} + \sigma_n^2\mathbf{I}_N)^{-1}\mathbf{k}_*$ where $\mathbf{k}_* := [k(\mathbf{z},\mathbf{z}^{(j)})]_{j=1}^N \in \mathbb{R}^N$ and $[\mathbf{G}]_{jl} := k(\mathbf{z}^{(j)},\mathbf{z}^{(l)})$ is the Gram matrix. To capture cross-correlations between disturbance components, essential for rotated uncertainty ellipsoids, we employ the Linear Model of Coregionalization (LMC), which yields full posterior covariance $\hat{\boldsymbol{\Sigma}}_w(\mathbf{z}) \in \mathbb{S}^n_{++}$ via mixing matrices; see~\cite{alvarezKernelsVectorValuedFunctions2012}.

\subsubsection{Credible Ellipsoid Decomposition}

The $(1{-}\alpha)$ credible region for the GP posterior is the ellipsoid $\mathcal{E}(\mathbf{z}) = \{\mathbf{w} : (\mathbf{w} - \hat{\boldsymbol{\mu}}_w)^\top\hat{\boldsymbol{\Sigma}}_w^{-1}(\mathbf{w} - \hat{\boldsymbol{\mu}}_w) \leq \chi^2_{n,1-\alpha}\}$.

\begin{lemma}[Credible Ellipsoid Decomposition]
\label{lem:ellipsoid_decomp}
The credible ellipsoid admits the representation $\mathcal{E}(\mathbf{z}) = \hat{\boldsymbol{\mu}}_w \oplus c_{n,\alpha}\hat{\boldsymbol{\Sigma}}_w^{1/2}\mathcal{B}_2^n$ where $c_{n,\alpha} := \sqrt{\chi^2_{n,1-\alpha}}$ and $\mathcal{B}_2^n := \{\mathbf{w} \in \mathbb{R}^n : \|\mathbf{w}\|_2 \leq 1\}$ is the Euclidean unit ball.
\end{lemma}

The proof is in Appendix~\ref{app:ellipsoid}. We absorb systematic error $\hat{\boldsymbol{\mu}}_w$ into the nominal model, leaving zero-mean uncertainty.

\begin{definition}[GP-Induced Anisotropy Matrix]
\label{def:P_construction}
Given posterior covariance $\hat{\boldsymbol{\Sigma}}_w(\mathbf{z})$ and confidence level $\alpha \in (0,1)$, the anisotropy matrix is
\begin{equation}
\mathbf{P}(\mathbf{z}) := c_{n,\alpha}\,\hat{\boldsymbol{\Sigma}}_w(\mathbf{z})^{1/2}.
\label{eq:P_def}
\end{equation}
\end{definition}

With template $\bar{\mathbf{W}} = \mathcal{B}_2^n$, the disturbance set $\mathbb{W}(\mathbf{z}) = \mathbf{P}(\mathbf{z})\bar{\mathbf{W}}$ coincides with the centered credible ellipsoid.

\subsubsection{Regularity of the Learned Field}

The invariance theory requires uniform bounds and continuity of $\mathbf{P}(\cdot)$. These properties follow from kernel regularity and properties of the matrix square root.

\begin{lemma}[Boundedness]
\label{lem:P_bounded}
Under Assumption~\ref{asm:kernel}, there exist $0 < p_{\min} \leq p_{\max} < \infty$ such that $p_{\min}\mathbf{I}_n \preceq \mathbf{P}(\mathbf{z}) \preceq p_{\max}\mathbf{I}_n$ for all $\mathbf{z} \in \mathcal{Z}$. Explicitly, $p_{\max} = c_{n,\alpha}\sqrt{\bar{k}_0}$ and $p_{\min} = c_{n,\alpha}\sqrt{\underline{\lambda}}$ with
\begin{equation}
\underline{\lambda} := \frac{\sigma_n^2\,\underline{\sigma}_0^2}{\sigma_n^2 + N_{\mathrm{gp}}\bar{k}_0},
\label{eq:lambda_floor}
\end{equation}
where $\boldsymbol{\Sigma}_0(\mathbf{z}) \succeq \underline{\sigma}_0^2\mathbf{I}_n$ is the LMC prior covariance and $\bar{k}_0$ bounds its eigenvalues.
\end{lemma}

The proof is in Appendix~\ref{app:bounded}.

\begin{lemma}[Lipschitz Continuity]
\label{lem:P_lipschitz}
Under Assumption~\ref{asm:kernel}, the map $\mathbf{P}:\mathcal{Z} \to \mathbb{S}^n_{++}$ is Lipschitz: $\|\mathbf{P}(\mathbf{z}_1) - \mathbf{P}(\mathbf{z}_2)\|_F \leq L_P\|\mathbf{z}_1 - \mathbf{z}_2\|$ with $L_P = c_{n,\alpha} L_\Sigma/(2\sqrt{\underline{\lambda}})$, where $L_\Sigma$ is the Lipschitz constant of $\hat{\boldsymbol{\Sigma}}_w(\cdot)$ and $\underline{\lambda}$ is as in~\eqref{eq:lambda_floor}.
\end{lemma}

The proof is in Appendix~\ref{app:lipschitz}. Lemmas~\ref{lem:P_bounded}--\ref{lem:P_lipschitz} establish that the GP-induced field satisfies Assumption~\ref{asm:P_regularity}, closing the loop between the learning and invariance frameworks. The base template $\bar{\boldsymbol{\Omega}} = r\mathcal{B}_2^n$ is computed once via~\eqref{eq:template_inflation}; subsequent GP updates refine only $\mathbf{P}(\cdot)$, with invariance maintained by constraining this field rather than recomputing sets.

\subsubsection{Learning Pipeline}

Algorithm~\ref{alg:learning} summarizes the offline pipeline. The base template pair is computed once (lines 3--4); thereafter only $\mathbf{P}^{(q)}(\cdot)$ is refreshed per epoch.

\begin{algorithm}[!t]
\caption{Anisotropy Field Learning}
\label{alg:learning}
\begin{algorithmic}[1]
\Require Trajectory data $\mathcal{D}$, confidence $\alpha$, kernel $k$
\Ensure Template pair $(\bar{\mathbf{W}},\bar{\boldsymbol{\Omega}})$, anisotropy field $\mathbf{P}(\cdot)$
\State $\mathbf{w}^{(j)} \gets \mathbf{x}^{(j+1)} - \mathbf{A}\mathbf{x}^{(j)} - \mathbf{B}\mathbf{u}^{(j)}$ \Comment{disturbance extraction}
\State Fit GP hyperparameters via marginal likelihood; form $\hat{\boldsymbol{\Sigma}}_w(\cdot)$
\State $r \gets 1/(1-\gamma_{\mathrm{cl}})$; \ $\bar{\boldsymbol{\Omega}} \gets r\mathcal{B}_2^n$; \ $\bar{\mathbf{W}} \gets \mathcal{B}_2^n$
\State $p_{\max} \gets \sup_{\mathbf{z}\in\mathcal{Z}}\|\mathbf{P}(\mathbf{z})\|_2$, \ $p_{\min} \gets \inf_{\mathbf{z}\in\mathcal{Z}}\lambda_{\min}(\mathbf{P}(\mathbf{z}))$ \Comment{Lemma~\ref{lem:P_bounded}}
\State $\mathbf{P}(\mathbf{z}) \gets c_{n,\alpha}\hat{\boldsymbol{\Sigma}}_w(\mathbf{z})^{1/2}$ \Comment{anisotropy field}
\State \Return $(\bar{\mathbf{W}}, \bar{\boldsymbol{\Omega}}, r),\ \mathbf{P}(\cdot)$
\end{algorithmic}
\end{algorithm}

\begin{proposition}[Learned Field Properties]
\label{prop:learned_field}
The field $\mathbf{P}(\mathbf{z})$ from Algorithm~\ref{alg:learning} satisfies: \textbf{(i)} $\mathbf{P}(\mathbf{z}) \in \mathbb{S}^n_{++}$; \textbf{(ii)} uniform bounds per Lemma~\ref{lem:P_bounded}; \textbf{(iii)} Lipschitz continuity per Lemma~\ref{lem:P_lipschitz}; \textbf{(iv)} $\mathbf{P}(\mathbf{z})\bar{\mathbf{W}}$ contains the true disturbance with probability $\geq 1-\alpha$.
\end{proposition}

\begin{proof}
Items (i)--(iii) restate~\eqref{eq:P_def}, Lemma~\ref{lem:P_bounded}, and Lemma~\ref{lem:P_lipschitz}. Item (iv) follows from Lemma~\ref{lem:ellipsoid_decomp}: $\mathbf{P}(\mathbf{z})\bar{\mathbf{W}} = c_{n,\alpha}\hat{\boldsymbol{\Sigma}}_w^{1/2}\mathcal{B}_2^n$ is the centered $(1{-}\alpha)$ credible ellipsoid, which contains $\mathbf{w}(\mathbf{z})$ with the stated probability under the GP posterior.
\end{proof}

\section{Invariance Conditions for Anisotropic Tubes}
\label{sec:invariance}

This section derives sufficient conditions for the tube field $\boldsymbol{\Omega}(\mathbf{z}) = \mathbf{P}(\mathbf{z})\bar{\boldsymbol{\Omega}}$ to be locally RPI, formulated as LMIs coupling the anisotropy field to closed-loop dynamics.

\begin{definition}[Gauge Function]
\label{def:gauge}
For a compact convex set $\mathbf{S} \subset \mathbb{R}^n$ with $\mathbf{0} \in \mathrm{int}(\mathbf{S})$, the \emph{gauge function} $\gamma_{\mathbf{S}}: \mathbb{R}^n \to \mathbb{R}_{\geq 0}$ is $\gamma_{\mathbf{S}}(\mathbf{w}) := \inf\{\lambda \geq 0 : \mathbf{w} \in \lambda \mathbf{S}\}$.
\end{definition}

\begin{lemma}[Gauge Function Properties]
\label{lem:gauge_props}
Let $\mathbf{S}$ be compact, convex with $\mathbf{0} \in \mathrm{int}(\mathbf{S})$. Then: \textbf{(i)} $\mathbf{S} = \{\mathbf{w} : \gamma_{\mathbf{S}}(\mathbf{w}) \leq 1\}$; \textbf{(ii)} $\gamma_{\mathbf{S}}(\alpha\mathbf{w}) = \alpha\gamma_{\mathbf{S}}(\mathbf{w})$ for $\alpha \geq 0$; \textbf{(iii)} $\gamma_{\mathbf{S}}(\mathbf{w}_1 + \mathbf{w}_2) \leq \gamma_{\mathbf{S}}(\mathbf{w}_1) + \gamma_{\mathbf{S}}(\mathbf{w}_2)$; \textbf{(iv)} $\gamma_{\mathbf{M}\mathbf{S}}(\mathbf{w}) = \gamma_{\mathbf{S}}(\mathbf{M}^{-1}\mathbf{w})$ for invertible $\mathbf{M}$.
\end{lemma}

Proofs of (i)--(iv) are in Appendix~\ref{app:gauge}. For the unit ball $\bar{\mathbf{W}} = \mathcal{B}_2^n$, the gauge reduces to the Euclidean norm: $\gamma_{\mathcal{B}_2^n}(\mathbf{w}) = \|\mathbf{w}\|_2$. For the transformed ellipsoid $\mathbf{P}\mathcal{B}_2^n$ with $\mathbf{P} \succ \mathbf{0}$, property (iv) yields $\gamma_{\mathbf{P}\mathcal{B}_2^n}(\mathbf{w}) = \|\mathbf{P}^{-1}\mathbf{w}\|_2$, hence $\mathbf{w} \in \mathbf{P}\mathcal{B}_2^n \Leftrightarrow \mathbf{w}^\top\mathbf{P}^{-2}\mathbf{w} \leq 1$. With the anisotropy field $\mathbf{P}:\mathcal{Z} \to \mathbb{S}^n_{++}$ from Section~\ref{sec:learning}, define the tube cross-section $\boldsymbol{\Omega}(\mathbf{z}) := \mathbf{P}(\mathbf{z})\bar{\boldsymbol{\Omega}}$ and disturbance set $\mathbb{W}(\mathbf{z}) := \mathbf{P}(\mathbf{z})\bar{\mathbf{W}}$. The field is \emph{locally RPI} if $\mathbf{e}_k \in \boldsymbol{\Omega}(\mathbf{z}_k)$ and $\mathbf{w}_k \in \mathbb{W}(\mathbf{z}_k)$ imply $\mathbf{e}_{k+1} \in \boldsymbol{\Omega}(\mathbf{z}_{k+1})$.

We work in \emph{template coordinates}: if $\mathbf{e}_k \in \mathbf{P}_k\bar{\boldsymbol{\Omega}}$, write $\mathbf{e}_k = \mathbf{P}_k\boldsymbol{\xi}$ with $\boldsymbol{\xi} \in \bar{\boldsymbol{\Omega}}$. The matrix $\mathbf{P}'^{-1}\mathbf{A}_{\mathrm{cl}}\mathbf{P}$ is the closed-loop error map in successor ellipsoid coordinates, while $\mathbf{P}'^{-1}\mathbf{P}$ is the relative disturbance injection. Invariance becomes a fixed-template containment test.

\begin{theorem}[Anisotropic RPI Condition]
\label{thm:aniso_rpi}
Let $\mathbf{P}(\mathbf{z})$ satisfy Assumption~\ref{asm:P_regularity}, and let $\bar{\boldsymbol{\Omega}}$ be RPI for $(\mathbf{A}_{\mathrm{cl}}, \bar{\mathbf{W}})$ with $\rho_{\bar{\boldsymbol{\Omega}}} < 1$. If for all transitions $\mathbf{z} \to \mathbf{z}'$
\begin{equation}
\mathbf{P}'^{-1}\mathbf{A}_{\mathrm{cl}}\mathbf{P}\,\bar{\boldsymbol{\Omega}} \oplus \mathbf{P}'^{-1}\mathbf{P}\,\bar{\mathbf{W}} \subseteq \bar{\boldsymbol{\Omega}},
\label{eq:aniso_rpi_cond}
\end{equation}
where $\mathbf{P} := \mathbf{P}(\mathbf{z})$ and $\mathbf{P}' := \mathbf{P}(\mathbf{z}')$, then $\boldsymbol{\Omega}(\mathbf{z}) = \mathbf{P}(\mathbf{z})\bar{\boldsymbol{\Omega}}$ is locally RPI.
\end{theorem}

The proof is in Appendix~\ref{app:aniso_rpi}. With templates $\bar{\boldsymbol{\Omega}} = r\mathcal{B}_2^n$, $\bar{\mathbf{W}} = \mathcal{B}_2^n$, condition~\eqref{eq:aniso_rpi_cond} reduces (after dividing by $r$) to $\|\boldsymbol{\xi}\|_2 \leq 1$, $\|\boldsymbol{\eta}\|_2 \leq 1 \Rightarrow \|\mathbf{M}\boldsymbol{\xi} + \mathbf{N}\boldsymbol{\eta}\|_2 \leq 1$, where $\mathbf{M} := \mathbf{P}'^{-1}\mathbf{A}_{\mathrm{cl}}\mathbf{P}$ and $\mathbf{N} := \tfrac{1}{r}\mathbf{P}'^{-1}\mathbf{P}$. The factor $1/r$ accounts for template headroom. We certify the quadratic implication by the sufficient S-procedure below.

\begin{theorem}[LMI Condition for Anisotropic RPI]
\label{thm:lmi_rpi}
The field $\boldsymbol{\Omega}(\mathbf{z})$ is locally RPI if for all transitions there exist $\lambda_1,\lambda_2 \geq 0$ with $\lambda_1+\lambda_2\leq 1$ such that
\begin{equation}
\begin{bmatrix} \lambda_1\mathbf{I}_n - \mathbf{M}^\top\mathbf{M} & -\mathbf{M}^\top\mathbf{N} \\ -\mathbf{N}^\top\mathbf{M} & \lambda_2\mathbf{I}_n - \mathbf{N}^\top\mathbf{N} \end{bmatrix} \succeq \mathbf{0}.
\label{eq:main_lmi}
\end{equation}
\end{theorem}

\begin{proof}
Let $\boldsymbol{\zeta}=[\boldsymbol{\xi}^\top,\boldsymbol{\eta}^\top]^\top$, $f_0=1-\|\mathbf{M}\boldsymbol{\xi}+\mathbf{N}\boldsymbol{\eta}\|_2^2$, $f_1=1-\|\boldsymbol{\xi}\|_2^2$, and $f_2=1-\|\boldsymbol{\eta}\|_2^2$. The sufficient S-procedure condition is $f_0-\lambda_1 f_1-\lambda_2 f_2\ge0$ for all $\boldsymbol{\zeta}$, with $\lambda_1,\lambda_2\ge0$. Expanding gives constant term $1-\lambda_1-\lambda_2$ and quadratic matrix~\eqref{eq:main_lmi}. The constraint $\lambda_1+\lambda_2\le1$ makes the constant term nonnegative, and~\eqref{eq:main_lmi} makes the quadratic term nonnegative. Hence $f_1\ge0$, $f_2\ge0\Rightarrow f_0\ge0$. No losslessness claim is made.
\end{proof}

\begin{corollary}[Spectral Sufficient Condition]
\label{cor:spectral}
Local RPI holds if $\|\mathbf{P}'^{-1}\mathbf{A}_{\mathrm{cl}}\mathbf{P}\|_2 + \tfrac{1}{r}\|\mathbf{P}'^{-1}\mathbf{P}\|_2 \leq 1$ for all transitions. With uniform bounds from Assumption~\ref{asm:P_regularity}, this is implied by $(\gamma_{\mathrm{cl}} + 1/r)\,\kappa_P \leq 1$ where $\kappa_P := p_{\max}/p_{\min}$ is the condition number of the anisotropy field.
\end{corollary}

\begin{remark}[Uniform vs. arc-wise verification]
\label{rem:uniform}
Since $\kappa_P \geq 1$, the uniform condition requires $\gamma_{\mathrm{cl}}+1/r \leq 1$, i.e. $r \geq 1/(1-\gamma_{\mathrm{cl}})$, exactly the template inflation factor. It is satisfiable with $\kappa_P = 1$ (isotropic fields) but becomes restrictive for $\kappa_P \gg 1$. The operative criterion is the arc-wise condition of Theorem~\ref{thm:graph_invariance}, which uses local ratios $p_i^+/p_j^-$; on a fine grid these approach unity even when the global condition number is large.
\end{remark}

The conditions of Theorems~\ref{thm:aniso_rpi} and~\ref{thm:lmi_rpi} require verification over all transitions $(\mathbf{x},\mathbf{u}) \to (\mathbf{x}',\mathbf{u}')$ in the continuous domain $\mathcal{Z}$. Since this is an uncountable set, direct verification is intractable. We discretize the operating space into a finite graph and derive sufficient conditions that, when satisfied at graph nodes, imply invariance over the entire domain.

\subsubsection{Graph Structure}

Let $\mathcal{G} = (\mathcal{V}, \mathcal{A})$ denote a directed graph where the vertex set $\mathcal{V} = \{1, \ldots, N_v\}$ indexes a partition of $\mathcal{Z}$ into regions $\{\mathcal{R}_i\}_{i=1}^{N_v}$ satisfying $\bigcup_{i} \mathcal{R}_i = \mathcal{Z}$ and $\mathrm{int}(\mathcal{R}_i) \cap \mathrm{int}(\mathcal{R}_j) = \emptyset$ for $i \neq j$. Each node $i$ has a representative point $\mathbf{z}_i = (\mathbf{x}_i, \mathbf{u}_i) \in \mathcal{R}_i$, typically the centroid. The arc set $\mathcal{A} \subseteq \mathcal{V} \times \mathcal{V}$ is built conservatively: consecutive rollout transitions seed arcs, and spectral reachable-set over-approximations add any successor regions reachable under $\mathbb{W}(\mathbf{z})$. Thus $(i,j) \in \mathcal{A}$ whenever there exist $\mathbf{z} \in \mathcal{R}_i$, $\mathbf{w} \in \mathbb{W}(\mathbf{z})$, and $\mathbf{v}' \in \mathbb{V}$ such that $\mathbf{x}' = \mathbf{A}_{\mathrm{cl}}\mathbf{x} + \mathbf{B}\mathbf{v} + \mathbf{w}$ satisfies $(\mathbf{x}', \mathbf{K}\mathbf{x}' + \mathbf{v}') \in \mathcal{R}_j$. In the implementation, $\mathcal{V}$ comes from $k$-means clustering of sampled operating points; spurious arcs only strengthen verification, since every added arc must also pass the invariance test.

\begin{assumption}[Graph Completeness]
\label{asm:graph_complete}
The graph $\mathcal{G}$ captures all dynamically feasible transitions: for any trajectory $\{(\mathbf{x}_k, \mathbf{u}_k)\}_{k \geq 0}$ of~\eqref{eq:dynamics} with $\mathbf{w}_k \in \mathbb{W}(\mathbf{x}_k, \mathbf{u}_k)$, if $(\mathbf{x}_k, \mathbf{u}_k) \in \mathcal{R}_i$ and $(\mathbf{x}_{k+1}, \mathbf{u}_{k+1}) \in \mathcal{R}_j$, then $(i,j) \in \mathcal{A}$.
\end{assumption}

\subsubsection{Spectral Bounds}

At each node $i$, store $\mathbf{P}_i := \mathbf{P}(\mathbf{z}_i)$. To account for variation of $\mathbf{P}(\cdot)$ within each region, we derive spectral bounds via Lipschitz continuity.

\begin{definition}[Region Diameter and Lipschitz Deviation]
\label{def:region_deviation}
For each region $\mathcal{R}_i$, the \emph{diameter} is $\mathrm{diam}(\mathcal{R}_i) := \sup_{\mathbf{z}, \mathbf{z}' \in \mathcal{R}_i} \|\mathbf{z} - \mathbf{z}'\|_2$ and the \emph{Lipschitz deviation bound} is $\delta_i := L_P \cdot \mathrm{diam}(\mathcal{R}_i)$.
\end{definition}

\begin{lemma}[Eigenvalue Bounds within Regions]
\label{lem:eigenvalue_bounds}
Let $\mathbf{P}(\cdot)$ satisfy Assumption~\ref{asm:P_regularity} with Lipschitz constant $L_P$. For any $\mathbf{z} \in \mathcal{R}_i$,
\begin{align}
\lambda_{\max}(\mathbf{P}(\mathbf{z})) &\leq \lambda_{\max}(\mathbf{P}_i) + \delta_i, \label{eq:lambda_max_bound} \\
\lambda_{\min}(\mathbf{P}(\mathbf{z})) &\geq \lambda_{\min}(\mathbf{P}_i) - \delta_i. \label{eq:lambda_min_bound}
\end{align}
\end{lemma}

The proof is in Appendix~\ref{app:eigenvalue_bounds}.

\begin{definition}[Spectral Bound Parameters]
\label{def:spectral_bounds}
For each node $i$, the upper spectral bound is $p_i^+ := \lambda_{\max}(\mathbf{P}_i) + \delta_i$ and the lower spectral bound is $p_i^- := \lambda_{\min}(\mathbf{P}_i) - \delta_i$.
\end{definition}

\begin{assumption}[Grid Fineness]
\label{asm:grid_fine}
The partition $\{\mathcal{R}_i\}$ is sufficiently fine such that $\delta_i < \lambda_{\min}(\mathbf{P}_i)$ for all $i \in \mathcal{V}$, ensuring $p_i^- > 0$.
\end{assumption}

\begin{lemma}[Spectral Norm Bounds for Matrix Products]
\label{lem:spectral_product_bounds}
For any $\mathbf{z} \in \mathcal{R}_i$ and $\mathbf{z}' \in \mathcal{R}_j$, the matrices $\mathbf{M} := \mathbf{P}(\mathbf{z}')^{-1}\mathbf{A}_{\mathrm{cl}}\mathbf{P}(\mathbf{z})$ and $\mathbf{N} := \tfrac{1}{r}\mathbf{P}(\mathbf{z}')^{-1}\mathbf{P}(\mathbf{z})$ satisfy
\begin{equation}
\|\mathbf{M}\|_2 \leq \frac{p_i^+}{p_j^-} \gamma_{\mathrm{cl}}, \qquad \|\mathbf{N}\|_2 \leq \frac{1}{r}\frac{p_i^+}{p_j^-}.
\label{eq:N_M_bounds}
\end{equation}
\end{lemma}

The proof is in Appendix~\ref{app:spectral_product}.

\begin{theorem}[Discrete Graph Invariance via Spectral Bounds]
\label{thm:graph_invariance}
Let $\mathcal{G} = (\mathcal{V}, \mathcal{A})$ satisfy Assumptions~\ref{asm:graph_complete}--\ref{asm:grid_fine}, with spectral bounds $\{p_i^+, p_i^-\}_{i \in \mathcal{V}}$ as in Definition~\ref{def:spectral_bounds}. If for every arc $(i,j) \in \mathcal{A}$
\begin{equation}
\frac{p_i^+}{p_j^-} \cdot \bigl(\gamma_{\mathrm{cl}} + \tfrac{1}{r}\bigr) \leq 1,
\label{eq:spectral_invariance_condition}
\end{equation}
then the tube field $\boldsymbol{\Omega}(\mathbf{z}) = \mathbf{P}(\mathbf{z})\bar{\boldsymbol{\Omega}}$ with $\bar{\boldsymbol{\Omega}} = r\mathcal{B}_2^n$ is locally RPI over $\mathcal{Z}$.
\end{theorem}

The proof is in Appendix~\ref{app:graph_invariance}.

\subsubsection{LMI Verification}

For applications requiring tighter bounds, we provide an LMI-based alternative with explicit conservatism quantification.

\begin{theorem}[LMI Verification with Robust Margin]
\label{thm:lmi_robust}
Let the spectral bounds $\{p_i^+, p_i^-\}$ be as in Definition~\ref{def:spectral_bounds}, and define $\bar{M}_{ij} := \frac{p_i^+}{p_j^-}\gamma_{\mathrm{cl}}$ and $\bar{N}_{ij} := \frac{1}{r}\frac{p_i^+}{p_j^-}$. If for every arc $(i,j) \in \mathcal{A}$ the LMI
\begin{equation}
\begin{bmatrix} \lambda_1\mathbf{I}_n - \bar{M}_{ij}^2 \mathbf{I}_n & -\bar{M}_{ij}\bar{N}_{ij}\mathbf{I}_n \\ -\bar{M}_{ij}\bar{N}_{ij}\mathbf{I}_n & \lambda_2\mathbf{I}_n - \bar{N}_{ij}^2 \mathbf{I}_n \end{bmatrix} \succeq \mathbf{0}
\label{eq:scalar_lmi}
\end{equation}
is feasible for some $\lambda_1,\lambda_2 \geq 0$ with $\lambda_1+\lambda_2\leq 1$, then the tube field is locally RPI.
\end{theorem}

\begin{proof}
The LMI~\eqref{eq:scalar_lmi} is Theorem~\ref{thm:lmi_rpi} with $\mathbf{M}^\top\mathbf{M}$, $\mathbf{N}^\top\mathbf{N}$, and $\mathbf{M}^\top\mathbf{N}$ replaced by scalar upper bounds. Since $\|\mathbf{M}\|_2 \leq \bar{M}_{ij}$ and $\|\mathbf{N}\|_2 \leq \bar{N}_{ij}$ imply $\mathbf{M}^\top\mathbf{M} \preceq \bar{M}_{ij}^2\mathbf{I}_n$ and $\mathbf{N}^\top\mathbf{N} \preceq \bar{N}_{ij}^2\mathbf{I}_n$, feasibility of~\eqref{eq:scalar_lmi} implies feasibility of~\eqref{eq:main_lmi} for all matrices satisfying the spectral bounds.
\end{proof}

For scalars, feasibility of~\eqref{eq:scalar_lmi} is equivalent to $\bar{M}_{ij}+\bar{N}_{ij}\leq 1$, recovering condition~\eqref{eq:spectral_invariance_condition}. Thus~\eqref{eq:spectral_invariance_condition} is the scalar upper-bound version of the full LMI. It is conservative because it uses the triangle inequality; in the scalar shorthand, the gap between $a+b\leq1$ and $a^2+b^2\leq1$ is at most $\sqrt{2}$.

\subsubsection{Existence and Ordering of Admissible Fields}

Define the \emph{admissible set} $\mathscr{A} := \{\mathbf{P}: \mathcal{Z} \to \mathbb{S}^n_{++} \mid \text{\eqref{eq:spectral_invariance_condition} holds for all } (i,j) \in \mathcal{A}\}$.

\begin{proposition}[Existence via Isotropic Fallback]
\label{prop:existence}
For any $p > 0$, the constant isotropic field $\mathbf{P}(\mathbf{z}) \equiv p\,\mathbf{I}_n$ satisfies~\eqref{eq:spectral_invariance_condition} for all $(i,j) \in \mathcal{A}$.
\end{proposition}

\begin{proof}
For a constant field, $\delta_i = 0$ for all $i$, so $p_i^+ = p_i^- = p$ and $p_i^+/p_j^- = 1$. The condition evaluates to $\gamma_{\mathrm{cl}} + 1/r = \gamma_{\mathrm{cl}} + (1 - \gamma_{\mathrm{cl}}) = 1 \leq 1$, with equality; the template inflation factor $r$ is precisely chosen to saturate this bound.
\end{proof}

The equality is not coincidental: $r = 1/(1-\gamma_{\mathrm{cl}})$ saturates the bound, so any non-constant field must extract its invariance margin from local regularity ($\delta_i \to 0$ on fine grids) rather than template headroom. The isotropic fallback guarantees $\mathscr{A} \neq \emptyset$ but captures no directional structure. The more interesting question is whether the GP-derived field is admissible.

\begin{corollary}[Admissibility of GP-Derived Field]
\label{thm:gp_sufficiency}
Let $\mathbf{P}^{\mathrm{GP}}(\mathbf{z}) = c_{n,\alpha}\hat{\boldsymbol{\Sigma}}_w(\mathbf{z})^{1/2}$ satisfy Assumption~\ref{asm:P_regularity}, and let $\mathcal{G} = (\mathcal{V}, \mathcal{A})$ satisfy Assumptions~\ref{asm:graph_complete}--\ref{asm:grid_fine}. If~\eqref{eq:spectral_invariance_condition} holds for all $(i,j) \in \mathcal{A}$ under $\mathbf{P}^{\mathrm{GP}}$, then $\boldsymbol{\Omega}(\mathbf{z}) = \mathbf{P}^{\mathrm{GP}}(\mathbf{z})\bar{\boldsymbol{\Omega}}$ is locally RPI.
\end{corollary}

The admissible set carries a natural partial order under the L\"owner ordering, with smaller fields yielding tighter tubes.

\begin{proposition}[Ordering of Admissible Fields]
\label{prop:field_ordering}
Let $\boldsymbol{\Sigma}_1, \boldsymbol{\Sigma}_2: \mathcal{Z} \to \mathbb{S}^n_{++}$ with $\boldsymbol{\Sigma}_1(\mathbf{z}) \preceq \boldsymbol{\Sigma}_2(\mathbf{z})$ pointwise, and let $\mathbf{P}_i := c\,\boldsymbol{\Sigma}_i^{1/2}$ for a scalar $c > 0$. If both $\mathbf{P}_1, \mathbf{P}_2$ satisfy~\eqref{eq:spectral_invariance_condition}, then $\boldsymbol{\Omega}_1(\mathbf{z}) := \mathbf{P}_1(\mathbf{z})\bar{\boldsymbol{\Omega}} \subseteq \mathbf{P}_2(\mathbf{z})\bar{\boldsymbol{\Omega}} =: \boldsymbol{\Omega}_2(\mathbf{z})$ for all $\mathbf{z}\in\mathcal{Z}$, and both are locally RPI.
\end{proposition}

The proof is in Appendix~\ref{app:field_ordering}. The hypothesis is on $\boldsymbol{\Sigma}_1 \preceq \boldsymbol{\Sigma}_2$ rather than $\mathbf{P}_1 \preceq \mathbf{P}_2$: for generic SPD matrices, $\mathbf{P}_1 \preceq \mathbf{P}_2$ does not imply $\mathbf{P}_1\mathcal{B}_2^n \subseteq \mathbf{P}_2\mathcal{B}_2^n$, since $\mathbf{P}_2^{-1}\mathbf{P}_1$ is generically non-symmetric. The structure $\mathbf{P}_i = c\,\boldsymbol{\Sigma}_i^{1/2}$ supplies the missing algebra, and the GP variance reduction gives $\hat{\boldsymbol{\Sigma}}_w^{(q+1)} \preceq \hat{\boldsymbol{\Sigma}}_w^{(q)}$ directly.

\subsubsection{Monotone Tube Contraction}
\label{sec:monotone}

The preceding results yield the paper's central guarantee.

\begin{theorem}[Monotone Tube Contraction]
\label{thm:monotone}
With fixed GP hyperparameters and nested dictionaries $\mathcal{D}^{(q)}\subseteq\mathcal{D}^{(q+1)}$, three monotonicity properties hold for all $\mathbf{z}\in\mathcal{Z}$: \textbf{(i)}~$\hat{\boldsymbol{\Sigma}}_w^{(q+1)}(\mathbf{z}) \preceq \hat{\boldsymbol{\Sigma}}_w^{(q)}(\mathbf{z})$; \textbf{(ii)}~$\mathbf{P}^{(q+1)}(\mathbf{z}) \preceq \mathbf{P}^{(q)}(\mathbf{z})$; \textbf{(iii)}~$\boldsymbol{\Omega}^{(q+1)}(\mathbf{z}) \subseteq \boldsymbol{\Omega}^{(q)}(\mathbf{z})$.
\end{theorem}

The proof is in Appendix~\ref{app:monotone}. Part (i) is the variance-reduction step, proved three ways (Routes A--C). Since covariance shrinkage can reduce $p_i^+$ and $p_j^-$ at different rates, condition~\eqref{eq:spectral_invariance_condition} is re-checked after each GP update.

\begin{proposition}[Sufficient Condition for Invariance Preservation]
\label{prop:preservation}
If at epoch $q$ the uniform condition $\frac{p_{\max}^{(q)}}{p_{\min}^{(q)}}\cdot(\gamma_{\mathrm{cl}}+1/r)\le 1$ holds, with $p_{\max}^{(q)} = \max_i p_i^+$ and $p_{\min}^{(q)} = \min_j p_j^-$ the extreme node bounds, then the arc-wise condition~\eqref{eq:spectral_invariance_condition} is satisfied for all $(i,j)\in\mathcal{A}$.

\end{proposition}

The proof is in Appendix~\ref{app:preservation}. In practice the uniform condition often persists: the noise floor $\sigma_n^2$ keeps the denominator bounded away from zero, and as data accumulate in high-uncertainty regions, $p_{\max}^{(q)}$ typically decreases faster than $p_{\min}^{(q)}$.

\section{Computational Architecture and Complexity}
\label{sec:complexity}

\subsubsection{Two-Time-Scale Architecture}

The verification condition requires the successor $(\mathbf{x}',\mathbf{u}')$ to evaluate $\mathbf{P}'$, yet $\mathbf{x}'$ depends on $\mathbf{w}\in\mathbb{W}(\mathbf{x},\mathbf{u})$, a circular dependency. We resolve this by \emph{freezing} the GP-derived field $\mathbf{P}^{(q)}(\cdot)$ during each learning epoch $q$:
\begin{enumerate}
\item Train the GP to obtain $\hat{\boldsymbol{\Sigma}}_w^{(q)}$ and construct $\mathbf{P}^{(q)}=c_{n,\alpha}[\hat{\boldsymbol{\Sigma}}_w^{(q)}]^{1/2}$.
\item With $\mathbf{P}^{(q)}$ fixed, verify invariance via the graph-based spectral condition of Theorem~\ref{thm:graph_invariance}.
\item Collect new data during operation.
\item Increment $q$ and repeat.
\end{enumerate}
Unlike the lifted method of~\cite{ramadanLearningBasedShrinkingDisturbanceInvariant2025a}, which recomputes full set geometry at each epoch, only the anisotropy matrix field is updated. Online tube evaluation reduces to a single GP prediction per operating point ($O(N_{\mathrm{gp}}^2)$ with cached Cholesky factors, where $N_{\mathrm{gp}}$ is the capped training subset size), independent of the MPC horizon, graph size, and learning epoch count.

\subsubsection{Complexity Analysis}

Let $n_x$, $n_u$, $n_w$ denote the state, input, and disturbance dimensions, $N_{\mathrm{gp}}$ the (capped) GP training set size, $R$ the number of latent GPs in the LMC kernel, and $|\mathcal{V}|$, $|\mathcal{A}|$ the graph vertex and arc counts. For the quadrotor benchmark of Section~\ref{sec:simulations}: $n_x{=}12$, $n_u{=}4$, $n_w{=}3$, $R{=}3$, $N_{\mathrm{gp}}{\le}300$, $|\mathcal{V}|{=}30$, $|\mathcal{A}|{=}63$.

\begin{table}[!t]
\centering
\caption{Offline Phase (One-Time or Per-Epoch)}
\label{tab:offline}
\footnotesize
\begin{tabular}{p{2.05cm}p{2.55cm}p{2.6cm}}
\toprule
\textbf{Step} & \textbf{Cost} & \textbf{Notes} \\
\midrule
Template pair & $O(n_x^3)$ & One-time; DARE solve via Schur decomposition. For $n_x{=}12$: ${<}1$\,ms. \\
Controller synthesis & $O(n_x^3)$ & One-time; LQR gain $\mathbf{K}$, then $\mathbf{A}_{\mathrm{cl}}=\mathbf{A}+\mathbf{B}\mathbf{K}$. \\
GP training & $O(E_{\mathrm{gp}}\,N_{\mathrm{gp}}^3\,R\,n_w)$ & $E_{\mathrm{gp}}{=}500$ L-BFGS epochs; dominated by Cholesky factorization of the $N_{\mathrm{gp}}{\times}N_{\mathrm{gp}}$ kernel matrix. With $N_{\mathrm{gp}}{\le}300$: ${\sim}30$\,s on GPU. \\
Graph construction & $O(N_{\mathrm{gp}}\,|\mathcal{V}|\,d)$ & $k$-means in the $d$-dimensional partition features; arcs from rollout transitions. ${<}1$\,ms. \\
Graph verification & $O(|\mathcal{A}|\,n_w^3)$ & Per-arc spectral check; $63$ arcs of $3{\times}3$ eigendecompositions: ${<}1$\,ms total. \\
\bottomrule
\end{tabular}
\end{table}

\begin{table}[!t]
\centering
\caption{Online Phase (Per MPC Step)}
\label{tab:online}
\footnotesize
\begin{tabular}{p{2.05cm}p{2.55cm}p{2.6cm}}
\toprule
\textbf{Step} & \textbf{Cost} & \textbf{Notes} \\
\midrule
GP prediction & $O(N_{\mathrm{gp}}^2)$ & Cached Cholesky factor; prediction is a kernel-vector product. For $N_{\mathrm{gp}}{=}300$: ${\sim}0.1$\,ms. \\
Template evaluation & $O(n_w^3)$ & $\mathbf{P}(\mathbf{z}_*)=c_{n_w,\alpha}\hat{\boldsymbol{\Sigma}}_w(\mathbf{z}_*)^{1/2}$; a $3{\times}3$ square root is negligible. \\
Total per step & $O(N_{\mathrm{gp}}^2)$ & Independent of horizon $H$, graph size $|\mathcal{V}|$, and epoch $q$. \\
\bottomrule
\end{tabular}
\end{table}

\begin{table*}[!t]
\centering
\caption{Comparison with the Lifted-RPI Baseline~\cite{ramadanLearningBasedShrinkingDisturbanceInvariant2025a}}
\label{tab:lifted}
\footnotesize
\begin{tabular}{lcc}
\toprule
\textbf{Metric} & \textbf{Lifted-RPI} & \textbf{Ans\"{a}tze-RPI} \\
\midrule
Iterations to convergence & $57$ & $0$ (direct evaluation) \\
Per-iteration cost & ${\sim}0.07$\,s (GPU) & N/A \\
Total offline time & ${\sim}4.2$\,s (GPU) & ${\sim}30$\,s (GP training) \\
Online query cost & $O(n_{\mathrm{facets}}\,N_{\mathrm{gp}})$ & $O(N_{\mathrm{gp}}^2)$ \\
Online query time & seconds--minutes & ${\sim}0.1$\,ms \\
Set representation & polytope (growing) & ellipsoid, $\tfrac{n_w(n_w+1)}{2}$ params \\
\bottomrule
\end{tabular}
\end{table*}

GP training dominates the offline cost at $O(N_{\mathrm{gp}}^3)$ per epoch and stays tractable with the training set capped at $N_{\mathrm{gp}}\le300$. The template pair and controller are computed once and reused across epochs. Online, the graph is not traversed: each tube cross-section is one cached GP covariance query and one $n_w\times n_w$ square root, with no iteration and no LMI solve. The lifted-RPI polytope facet count grows combinatorially with dimension and iteration count, while the ans\"{a}tze parameterizes each cross-section with $n_w(n_w{+}1)/2 = 6$ Cholesky entries for $n_w=3$, regardless of iteration count; this fixed representation size is what permits the 12-state quadrotor.

Verification is offline and scales with the feature dimension of the partition, not the full state-input space. For the quadrotor, disturbances enter only the 3D velocity subspace, so the partition operates on disturbance-active features rather than all $n_x + n_u = 16$ coordinates. A dense grid in high dimensions remains subject to the curse of dimensionality, motivating disturbance-active subspaces, adaptive clustering, or sparse graph refinement.

\section{Simulation and Results}
\label{sec:simulations}

\subsection{Setup}
\label{sec:sim_setup}

We validate on a 3-D 12-state quadrotor tracking a 3-D Lissajous trajectory, linearized about hover, with velocity-subspace residual
\begin{equation}
\mathbf{w}=-\frac{\beta_1}{m}\|\mathbf{v}\|\mathbf{v}-\beta_2\mathbf{a}_{\mathrm{cmd}}+\sigma_w\boldsymbol{\varepsilon},
\qquad \boldsymbol{\varepsilon}\sim\mathcal{N}(\mathbf{0},\mathbf{I}_3),
\label{eq:residual}
\end{equation}
which captures quadratic drag, actuator inefficiency, and process noise ($\beta_1{=}0.15$, $\beta_2{=}0.08$, $\sigma_w{=}0.05$). The system is discretized at $T_s=0.02$\,s; $\rho(\mathbf{A}_{\mathrm{cl}})=0.983$. The multi-output GP uses LMC with $R{=}3$ latent GPs and SE kernels. The template pair uses $\gamma_{\mathrm{cl}}=0.99844$ and $r\approx641$. We compare four parameterizations: full $\mathbf{P}(\mathbf{z})$ (proposed), local isotropic $\sigma_{\max}(\mathbf{z})\mathbf{I}_3$~\cite{kohlerComputationallyEfficientRobust2021}, local diagonal $\mathrm{diag}(\hat{\sigma}_{w,i}(\mathbf{z}))$~\cite{solopertoLearningBasedRobustModel2018}, and a non-adaptive homothetic baseline using one fixed global template for all operating points~\cite{rakovicHomotheticTubeModel2012}. An interactive companion document\footnote{\url{https://thelastpixie.github.io/ansatz-rpi-v2/}} reproduces every figure of this section with expandable derivations, and embeds a real-time 3-D quadrotor simulator\footnote{\url{https://thelastpixie.github.io/ansatz-rpi-v2/sim.html}} with RK4 integration at 200\,Hz, the disturbance model~\eqref{eq:residual}, and live rendering of the anisotropic disturbance ellipsoids and RPI tube.

\subsection{Controller Design-Space Sweep}
\label{sec:sim_sweep}

The deployed LQR configuration yields the closed-loop figures above ($\rho(\mathbf{A}_{\mathrm{cl}})=0.983$, $\gamma_{\mathrm{cl}}=0.99844$, $r\approx641$). To characterize how sensitive certification is to controller tuning, a 7-layer adaptive sweep evaluated $50{,}365$ unique LQR weight configurations across the position, velocity, attitude, rate, thrust, and torque weights $(Q_p, Q_v, Q_a, Q_r, R_T, R_\tau)$: coarse factorial (L1), attitude/rate refinement (L2), fine zoom (L3), asymmetric axes (L4), Latin hypercube sampling in 8D (L5), fine polish (L6), and a dedicated $60{\times}60$ heatmap grid (L7). Each configuration is tested with the weighted-norm certificate of Section~\ref{sec:invariance} at three confidence levels:
$41{,}412$ configurations pass at $\alpha=0.90$ (82.2\%), $40{,}975$ at $\alpha=0.95$ (81.4\%), and $40{,}246$ at $\alpha=0.99$ (79.9\%).

The controller sweep has three deployment-relevant properties. First, the feasible region in weight space is broad: roughly four in five evaluated configurations certify at $\alpha=0.95$, so the invariance test does not force a narrow corner of the tuning space. Second, the certified norm varies continuously across the $\log_{10}(Q_p)\times\log_{10}(Q_v)$ heatmap, with a ridge of high certifiable confidence at low position and velocity weights, corresponding to controllers with aggressive disturbance rejection (Fig.~\ref{fig:pareto2}f). Third, faster contraction trades against margin: lower $\gamma$ correlates with a smaller gap to the unit threshold (Fig.~\ref{fig:pareto1}b). Layers L2 and L6 contribute the most passing configurations, consistent with the secondary weights $(Q_a, Q_r, R_\tau)$ having less influence on certification than the position and velocity weights.

\begin{figure*}[!t]
\centering
\subfloat[$\alpha^*$ vs.\ weighted norm]{\includegraphics[width=0.24\textwidth]{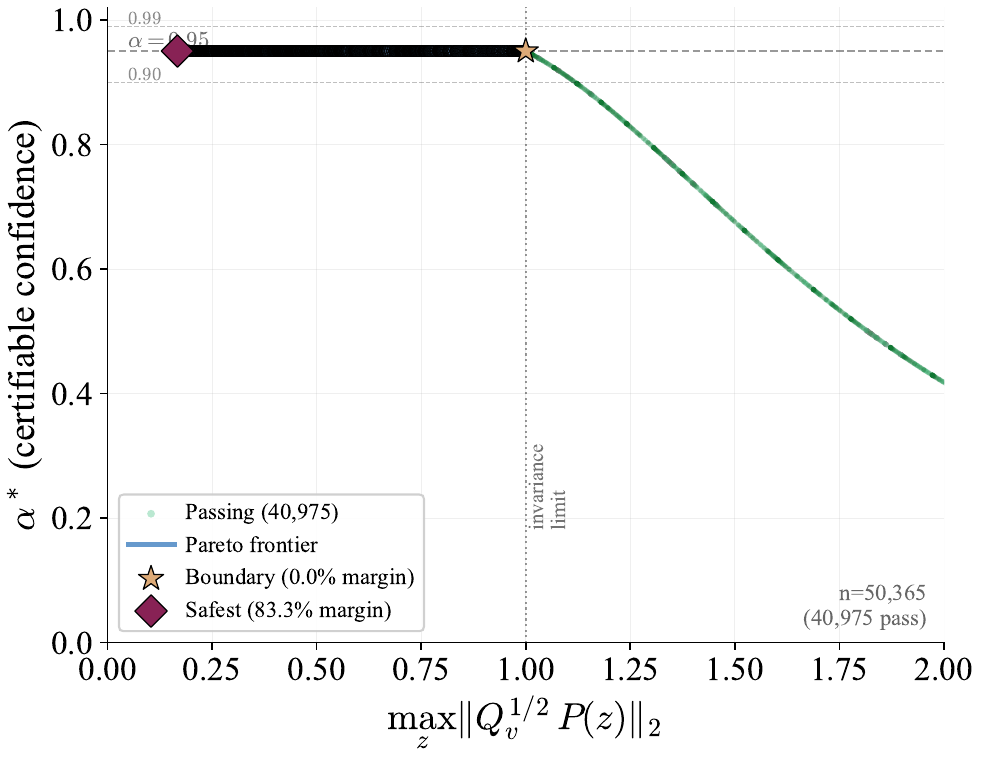}\label{fig:pareto1a}}\hfil
\subfloat[contraction rate vs.\ norm]{\includegraphics[width=0.24\textwidth]{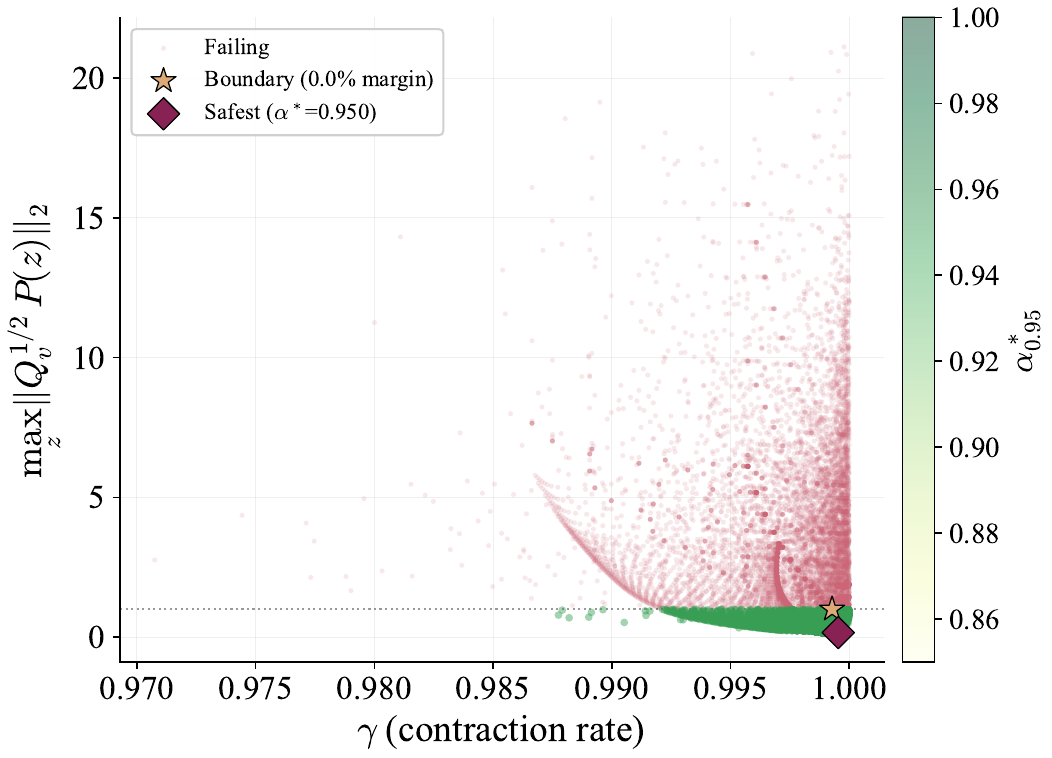}\label{fig:pareto1b}}\hfil
\subfloat[distribution of $\alpha^*$]{\includegraphics[width=0.24\textwidth]{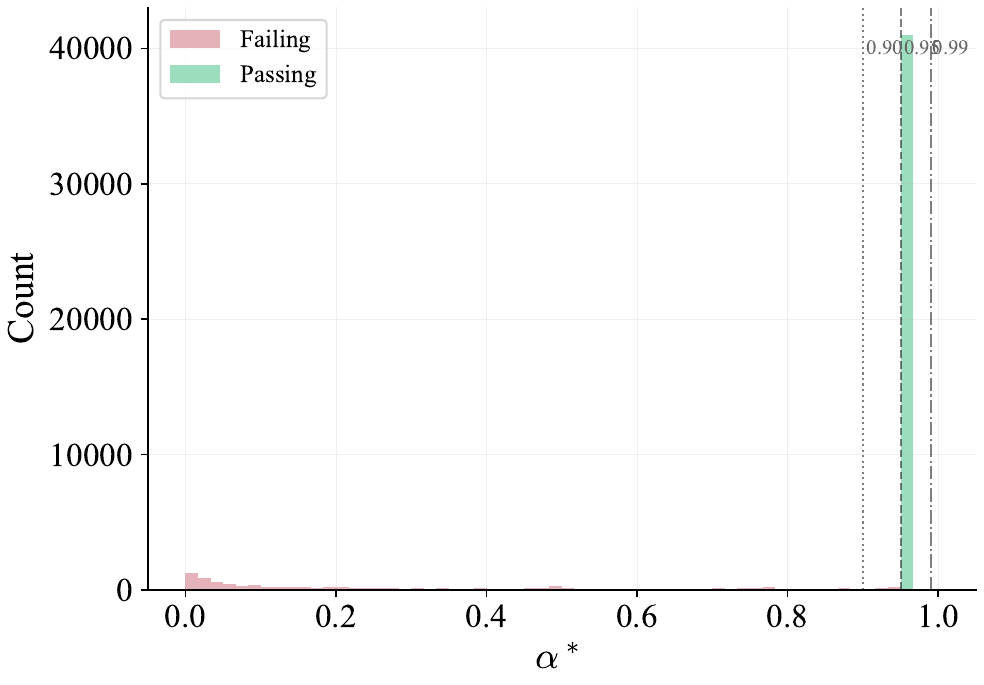}\label{fig:pareto1c}}\hfil
\subfloat[multi-$\alpha$ envelope]{\includegraphics[width=0.24\textwidth]{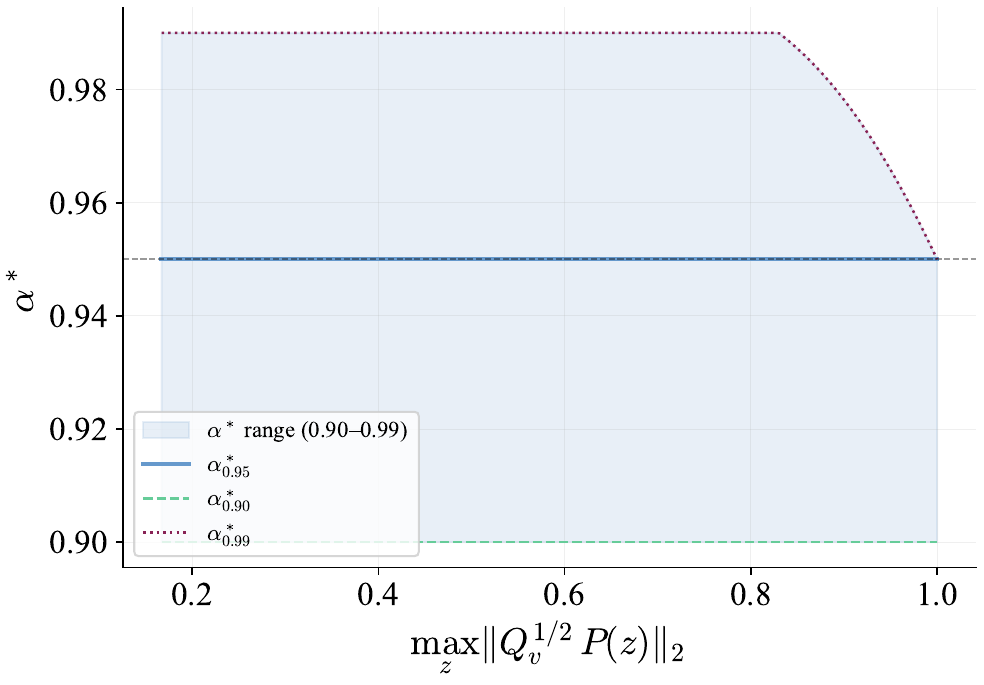}\label{fig:pareto1d}}
\caption{Controller sweep, part 1. (a)~The Pareto frontier traces the trade-off between certifiable confidence $\alpha^*$ and the weighted invariance norm; passing configurations cluster below the unit threshold. (b)~Contraction rate vs.\ norm, colored by $\alpha^*$: faster contraction correlates with smaller margin. (c)~Most evaluated configurations achieve $\alpha^* \geq 0.90$, with a spike at the sweep's primary target $\alpha^* = 0.95$. (d)~The confidence envelope separates cleanly across $\alpha \in \{0.90, 0.95, 0.99\}$ at higher norms, quantifying the cost of increased confidence.}
\label{fig:pareto1}
\end{figure*}

\begin{figure*}[!t]
\centering
\subfloat[sweep layer contributions]{\includegraphics[width=0.24\textwidth]{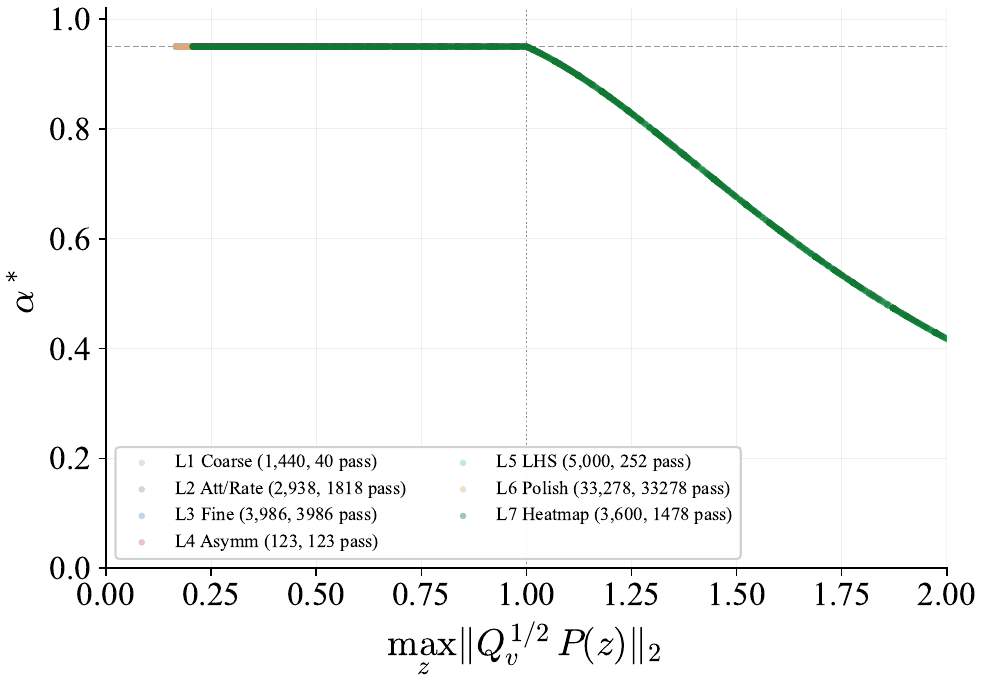}\label{fig:pareto2e}}\hfil
\subfloat[$\alpha^*$ heatmap over $Q_p \times Q_v$]{\includegraphics[width=0.24\textwidth]{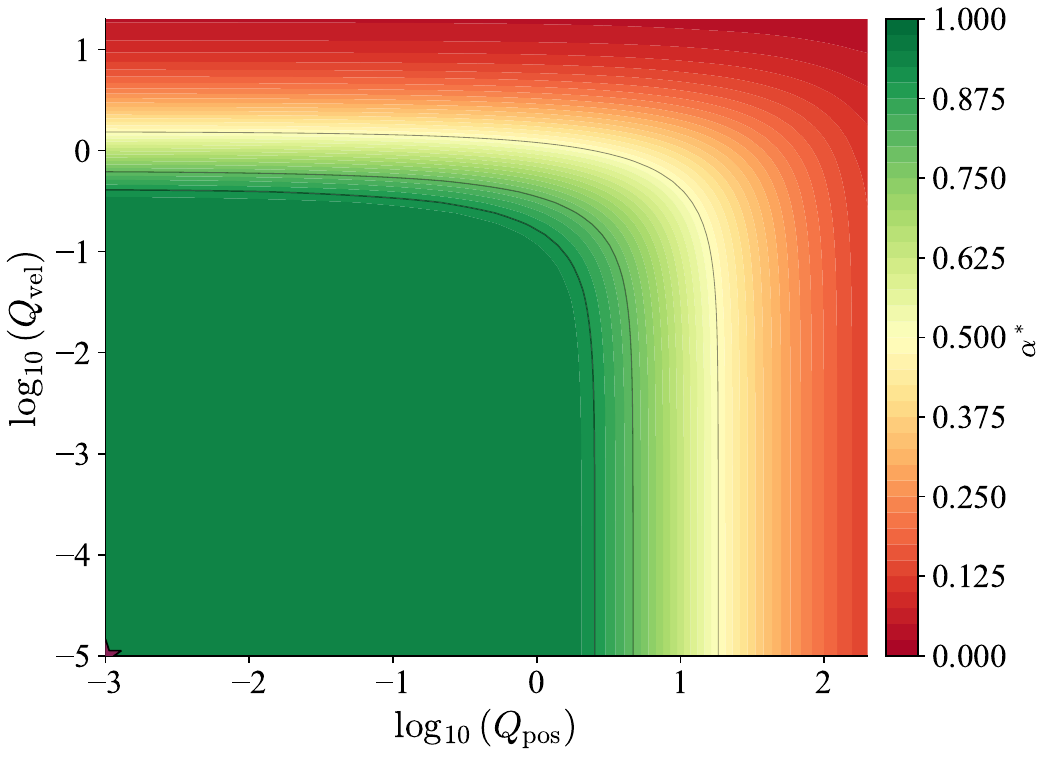}\label{fig:pareto2f}}\hfil
\subfloat[top-20 configurations]{\includegraphics[width=0.24\textwidth]{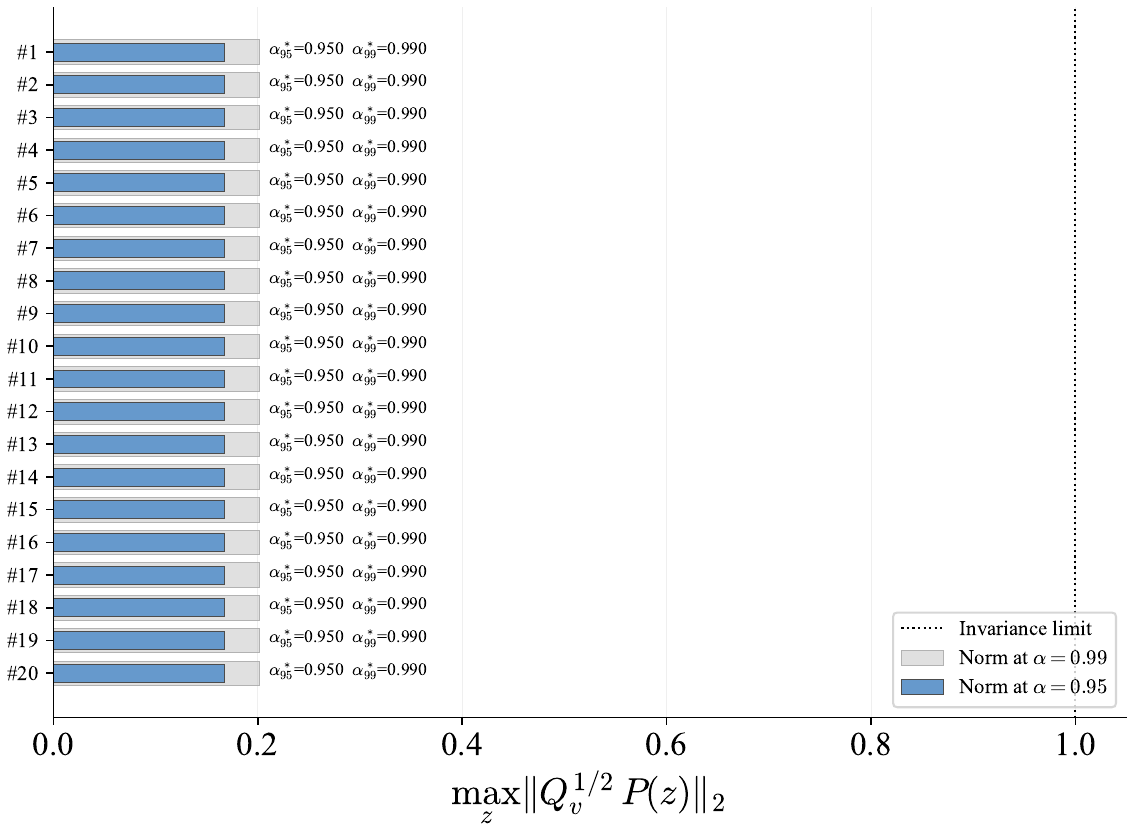}\label{fig:pareto2g}}\hfil
\subfloat[asymmetric $z$-weight effect]{\includegraphics[width=0.24\textwidth]{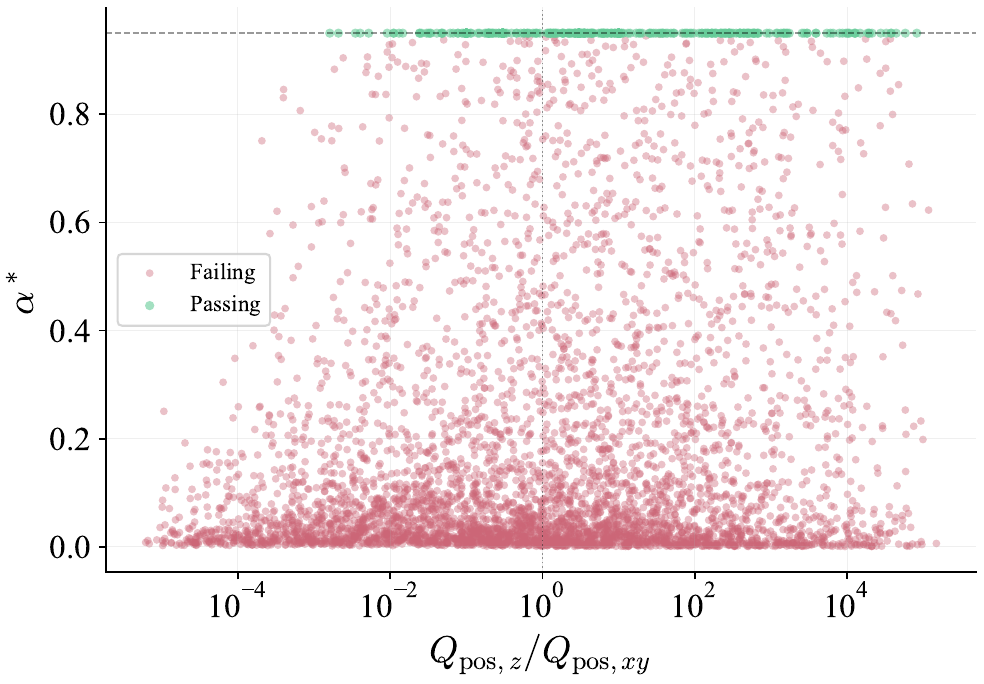}\label{fig:pareto2h}}
\caption{Controller sweep, part 2. (e)~Per-layer contributions in (norm, $\alpha^*$) space; L6 dominates volume with $33{,}278$ fine-grid evaluations, while L5 provides broad 8D coverage. (f)~The L7 heatmap over $\log_{10}(Q_p)\times\log_{10}(Q_v)$ with remaining weights fixed; a ridge of high $\alpha^*$ emerges at low $Q_p$ and $Q_v$. (g)~Top-20 passing configurations with norms at $\alpha=0.95$ and $\alpha=0.99$. (h)~Effect of asymmetric position weighting $Q_{p,z}/Q_{p,xy}$: near-isotropic ratios certify best, though moderate asymmetry remains feasible.}
\label{fig:pareto2}
\end{figure*}

\subsection{GP Learning: LMC vs.\ Independent Outputs}
\label{sec:sim_gp}

The LMC kernel ($R=3$ latent GPs) captures cross-correlations between velocity-axis disturbance components. Table~\ref{tab:gp_rmse} compares held-out RMSE against independent per-output GPs: LMC matches or improves accuracy per task while also producing the joint covariance $\hat{\boldsymbol{\Sigma}}_w(\mathbf{z})$ with off-diagonal terms, the ingredient that scalar and diagonal baselines lack. Figure~\ref{fig:lmc} plots the joint posteriors.

\begin{table}[!t]
\centering
\caption{GP Prediction Accuracy (RMSE)}
\label{tab:gp_rmse}
\footnotesize
\begin{tabular}{lccc}
\toprule
\textbf{Task} & \textbf{Independent GP} & \textbf{LMC ($R{=}3$)} & \textbf{Change} \\
\midrule
$w_{v_x}$ & $0.0477$ & $0.0480$ & $-0.6\%$ \\
$w_{v_y}$ & $0.0524$ & $0.0520$ & $+0.6\%$ \\
$w_{v_z}$ & $0.0594$ & $0.0588$ & $+1.1\%$ \\
\bottomrule
\end{tabular}
\end{table}

\begin{figure}[!t]
\centering
\includegraphics[trim={0 0 0 1.65cm}, clip=true,width=\columnwidth]{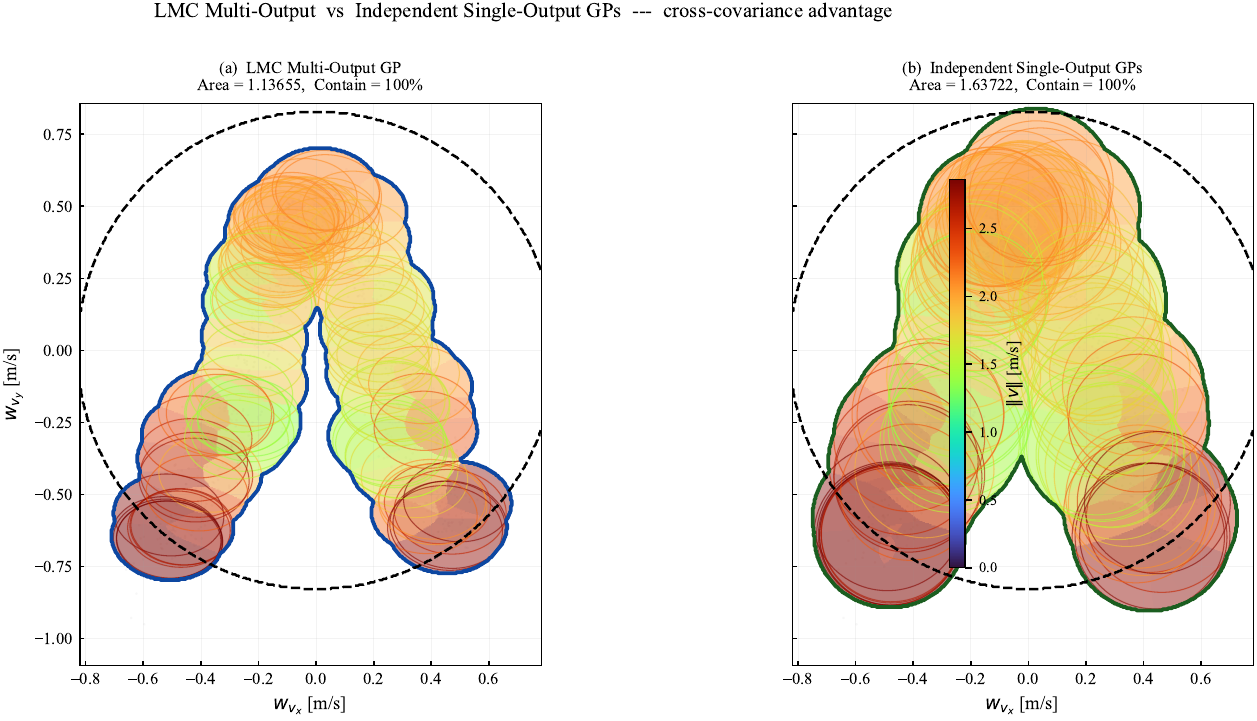}
\caption{LMC vs.\ independent GPs: the joint posterior captures inter-axis correlations used for anisotropic template construction.}
\label{fig:lmc}
\end{figure}

The weighted spectral certificate is evaluated at 100 test points along the closed-loop trajectory. All 100 points pass, with weighted norms in $[0.279, 0.639]$, well below the unit threshold. The margin is consistent with Remark~\ref{rem:uniform}: the certificate draws its slack from local ratios rather than the global condition number. The operating space is discretized into $|\mathcal{V}|=30$ $k$-means nodes with $|\mathcal{A}|=63$ arcs. The estimated Lipschitz constants are $L_P = 0.283$ for the tube field and $L_{P,w} = 0.342$ for the disturbance-set field. The per-node tube-sizing linear program yields $\alpha^*/r \in [0.605, 0.612]$ with all 30 nodes feasible, so the arc-wise condition~\eqref{eq:spectral_invariance_condition} holds across the verified graph. Theorem~\ref{thm:monotone} predicts nested tubes as data accumulate. Table~\ref{tab:epochs} reports four logged checkpoints as the training set grows from 80 to 400 points: $p_{\max}$ falls from $1.050$ to $0.194$ and the field condition number $\kappa_P$ from $7.10$ to $1.44$, while the per-point pass rate reaches $100/100$ from epoch~1 onward. The L\"owner ordering $\hat{\boldsymbol{\Sigma}}_w^{(q+1)}(\mathbf{z}) \preceq \hat{\boldsymbol{\Sigma}}_w^{(q)}(\mathbf{z})$ was verified numerically at every test point and epoch pair, with no violations.

\begin{table}[!t]
\centering
\caption{Monotone Tube Contraction Across Learning Epochs}
\label{tab:epochs}
\footnotesize
\begin{tabular}{ccccc}
\toprule
\textbf{Epoch} & $N_{\mathrm{train}}$ & $p_{\max}$ & $\kappa_P$ & \textbf{Pass} \\
\midrule
0 & 80  & 1.050 & 7.10 & 98/100 \\
1 & 160 & 0.383 & 2.76 & 100/100 \\
2 & 240 & 0.231 & 1.68 & 100/100 \\
4 & 400 & 0.194 & 1.44 & 100/100 \\
\bottomrule
\end{tabular}
\end{table}

\subsection{Comparison Against Scalar and Diagonal Baselines}
\label{sec:sim_comparison}

Table~\ref{tab:results} reports tube-volume ratios (baseline over proposed). The homothetic row denotes a non-adaptive baseline: one fixed global template is used at all operating points, rather than the state--input-dependent field $\mathbf{P}(\mathbf{z})$. It is therefore a Rakovi\'c-style homothetic comparison in transformation class, not the $\varepsilon$-mRPI polytope produced by the Rakovi\'c--Mayne invariant-set construction. The full-matrix parameterization achieves a $195\times$ volume reduction relative to this non-adaptive homothetic baseline in the 3D velocity subspace and $2.1{\times}10^5$ in the joint 7D velocity-control subspace $(\mathbf{e}_v,\,\delta\mathbf{u})$, where $\delta\mathbf{u}=\mathbf{K}\mathbf{e}\in\mathbb{R}^4$ captures the full control correction induced by the error. This illustrates \emph{dimensional compounding}: a modest per-axis overestimation compounds exponentially with dimension, so the homothetic 3D ratio of $195\times$ amplifies by ${\approx}10^3$ when the four control channels are included. Even the diagonal baseline grows from $4.3\times$ to $31.2\times$. All methods exceed $99\%$ empirical containment, consistent with the $95\%$ GP credible level; this is a coverage check, while the deterministic invariance proof is conditional on $\mathbf{w}\in\mathbb{W}(\mathbf{z})$.

\begin{table*}[!t]
\centering
\caption{Tube Volume Ratio (Baseline\,/\,Proposed) in the Velocity Subspace ($\mathbb{R}^3$) and the Joint Velocity-Control Subspace ($\mathbb{R}^7$)}
\label{tab:results}
\footnotesize
\begin{tabular}{lcccc}
\toprule
\textbf{Method} & \textbf{3D (vel)} & \textbf{7D (vel+ctrl)} & \textbf{Amp.} & \textbf{Cnt.} \\
\midrule
Proposed (full $\mathbf{P}$) & $1.0\times$ & $1.0\times$ & --- & 99.6\% \\
Isotropic~\cite{kohlerComputationallyEfficientRobust2021} & $1.3\times$ & $2.1\times$ & $1.6\times$ & 99.7\% \\
Diagonal~\cite{solopertoLearningBasedRobustModel2018} & $4.3\times$ & $31.2\times$ & $7.3\times$ & 99.8\% \\
Non-adaptive homothetic~\cite{rakovicHomotheticTubeModel2012} & $195\times$ & $2.1{\times}10^5$ & $1.1{\times}10^3$ & 99.5\% \\
\bottomrule
\end{tabular}
\end{table*}

\begin{remark}[Interpretation of the homothetic baseline] The homothetic baseline in Table~\ref{tab:results} is a single global template held fixed over the operating space. It is used to isolate the effect of replacing a non-adaptive template by the learned anisotropic field $\mathbf{P}(\mathbf{z})$. The reported $195\times$ and $2.1{\times}10^5$ ratios should therefore be read as comparisons against this fixed-template baseline. They are not direct volume ratios against a separately computed Rakovi\'c--Mayne $\varepsilon$-mRPI polytope. \end{remark}

Table~\ref{tab:areas} complements the volume ratios with 2D projection areas of the disturbance union and RPI error union in each velocity plane: the proposed sets are the tightest in every projection, and the homothetic RPI projections are an order of magnitude larger. Figure~\ref{fig:fourmethod} overlays the four methods in all three velocity planes; Fig.~\ref{fig:proj} adds control-space and cross-space projections, where directional coupling between velocity errors and control authority is visible; and Fig.~\ref{fig:threed} shows the 3D union boundaries extracted by marching cubes.

\begin{table}[!t]
\centering
\caption{2D Projection Areas (m$^2$/s$^2$)}
\label{tab:areas}
\footnotesize
\begin{tabular}{lcccc}
\toprule
\textbf{Projection} & \textbf{Proposed} & \textbf{Homoth.} & \textbf{Diag.} & \textbf{Isotr.} \\
\midrule
Dist.\ $v_x$--$v_y$ & $\mathbf{1.100}$ & $2.849$ & $1.636$ & $1.161$ \\
RPI $e_{v_x}$--$e_{v_y}$ & $\mathbf{4.278}$ & $77.51$ & $8.325$ & $4.911$ \\
RPI $e_{v_x}$--$e_{v_z}$ & $\mathbf{2.754}$ & $35.58$ & $5.476$ & $3.196$ \\
RPI $e_{v_y}$--$e_{v_z}$ & $\mathbf{2.684}$ & $48.84$ & $5.771$ & $3.540$ \\
\bottomrule
\end{tabular}
\end{table}

\begin{figure*}[!t]
\centering
\subfloat[dist.\ $v_x$--$v_y$]{\includegraphics[width=0.31\textwidth]{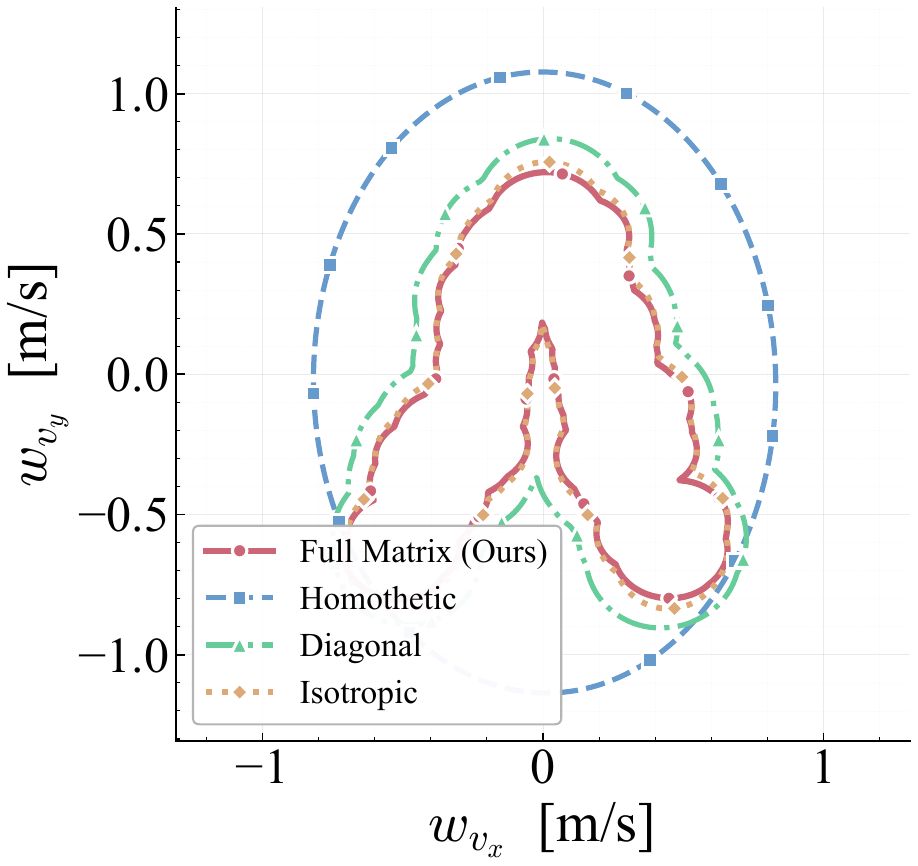}\label{fig:cmp_a}}\hfil
\subfloat[dist.\ $v_x$--$v_z$]{\includegraphics[width=0.31\textwidth]{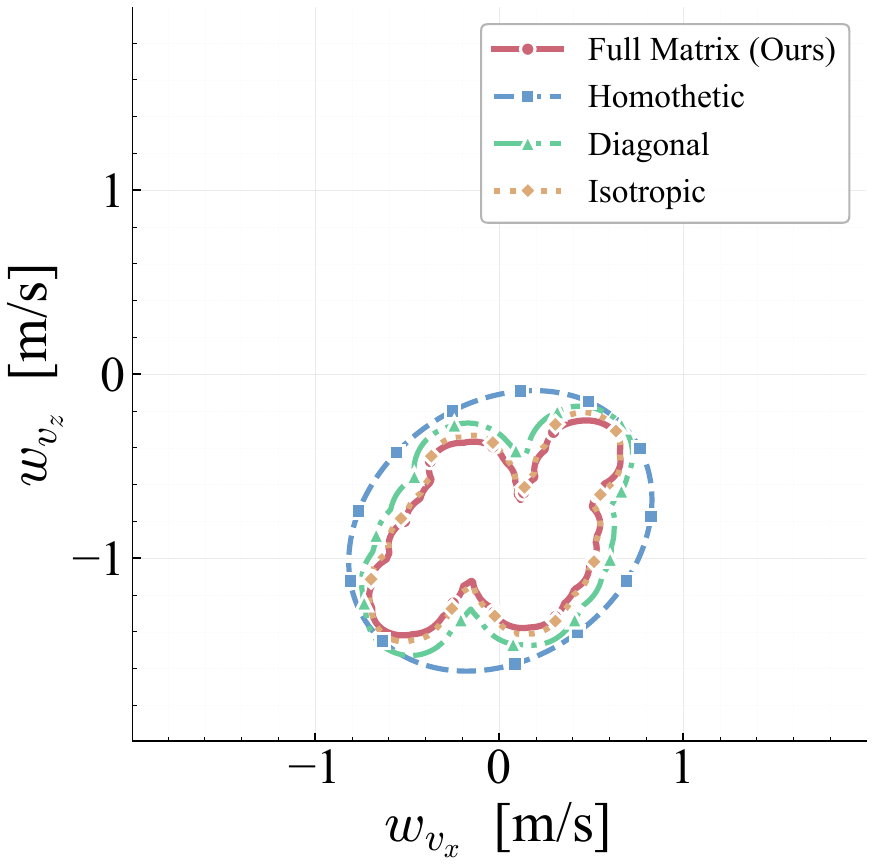}\label{fig:cmp_b}}\hfil
\subfloat[dist.\ $v_y$--$v_z$]{\includegraphics[width=0.31\textwidth]{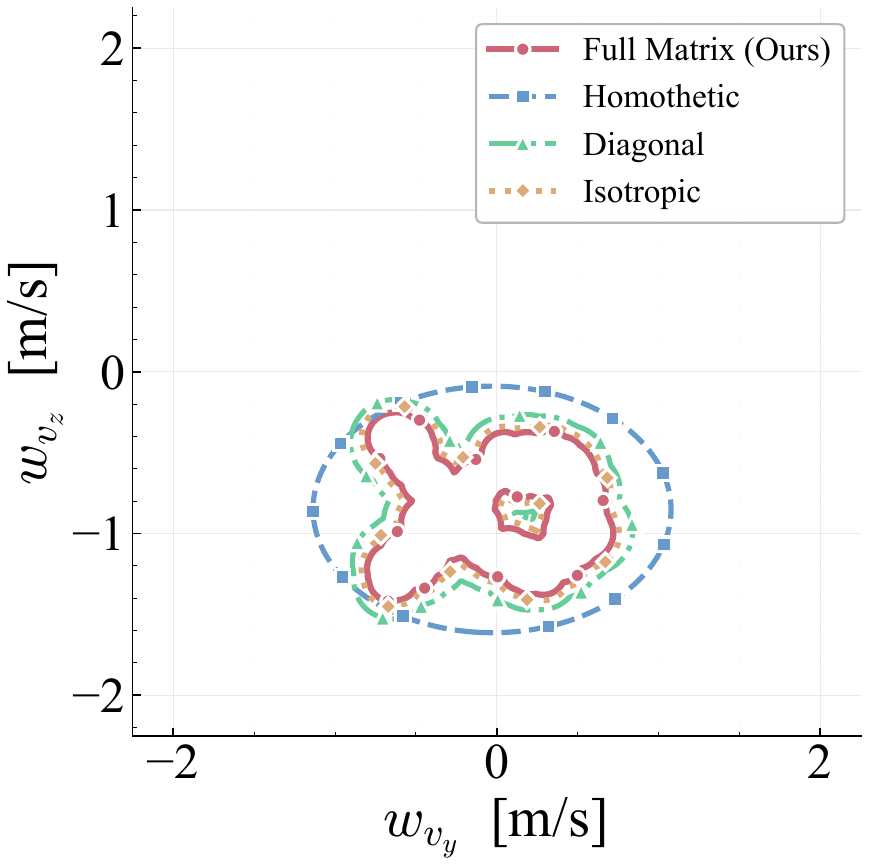}\label{fig:cmp_c}}\\
\subfloat[RPI $e_{v_x}$--$e_{v_y}$]{\includegraphics[width=0.31\textwidth]{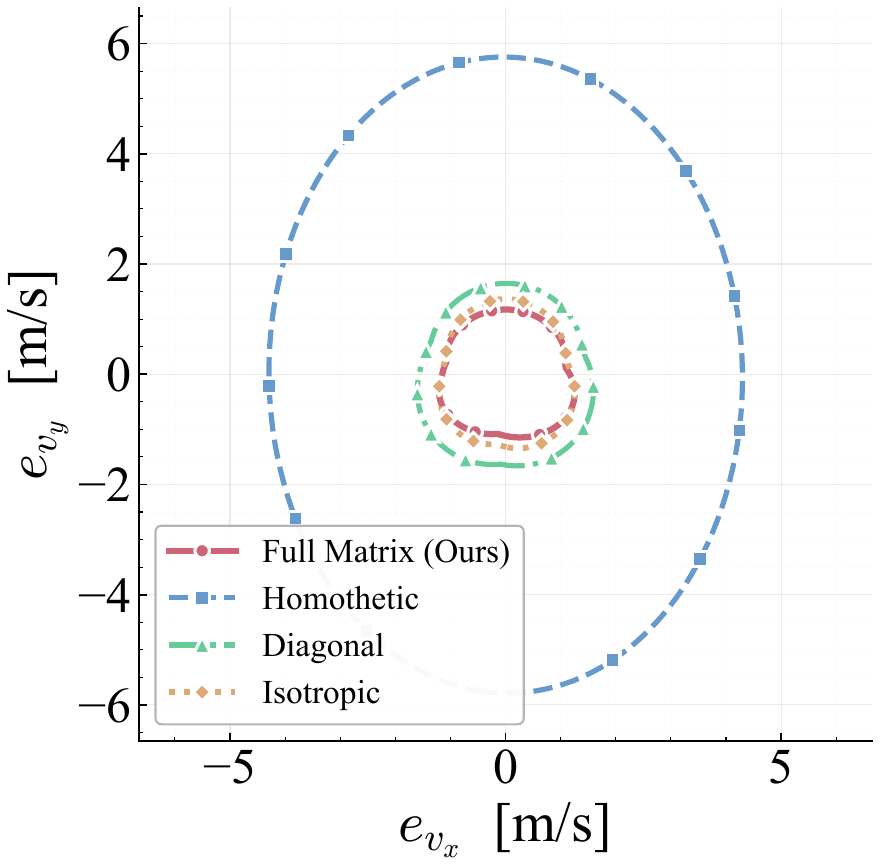}\label{fig:cmp_d}}\hfil
\subfloat[RPI $e_{v_x}$--$e_{v_z}$]{\includegraphics[width=0.31\textwidth]{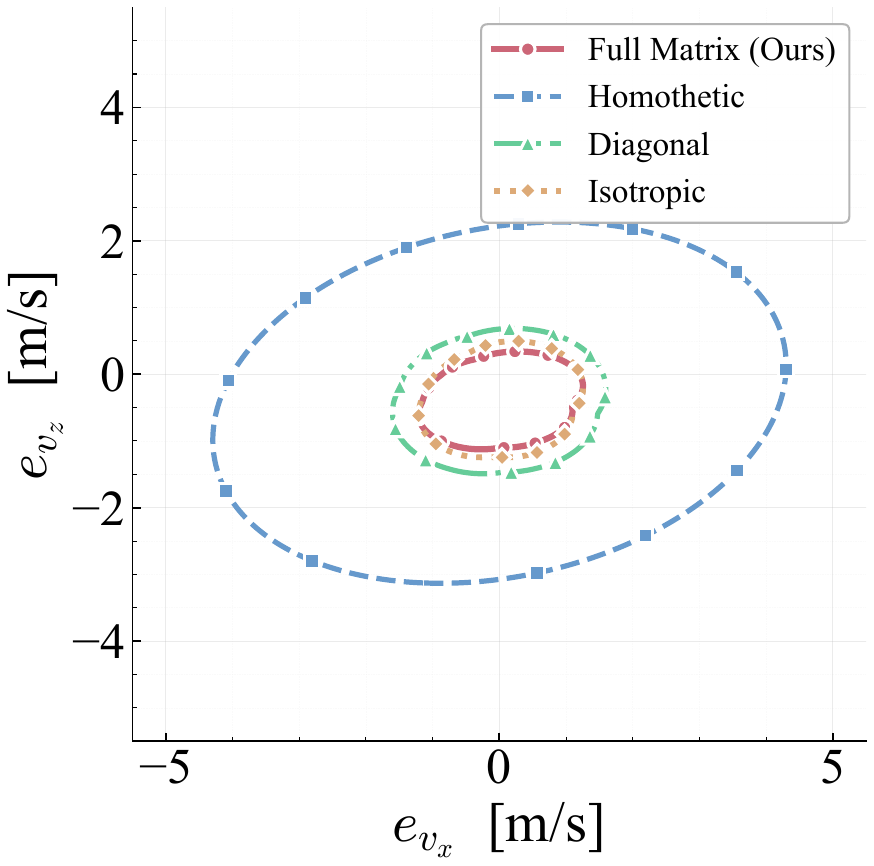}\label{fig:cmp_e}}\hfil
\subfloat[RPI $e_{v_y}$--$e_{v_z}$]{\includegraphics[width=0.31\textwidth]{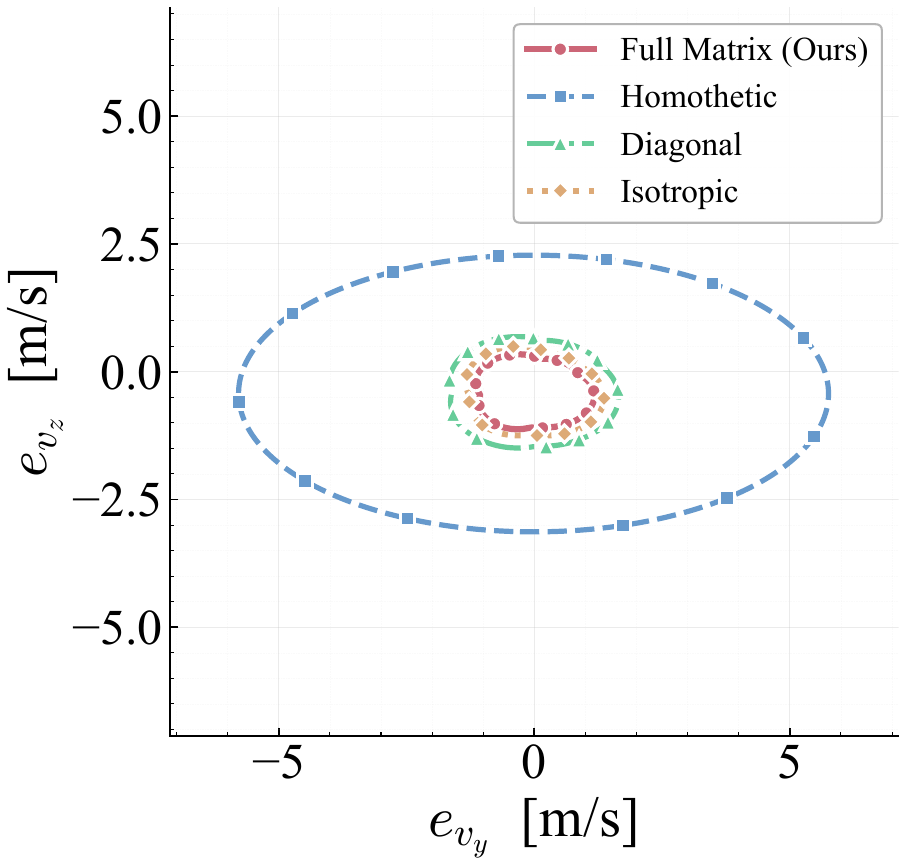}\label{fig:cmp_f}}
\caption{Four-method comparison in the three velocity planes: disturbance unions (top) and RPI error unions (bottom). The proposed anisotropic sets stay close to the operating region, while the homothetic baseline gives a much larger over-approximation.}
\label{fig:fourmethod}
\end{figure*}

\begin{figure*}[!t]
\centering
\subfloat[$v_x$--$v_y$]{\includegraphics[width=0.31\textwidth]{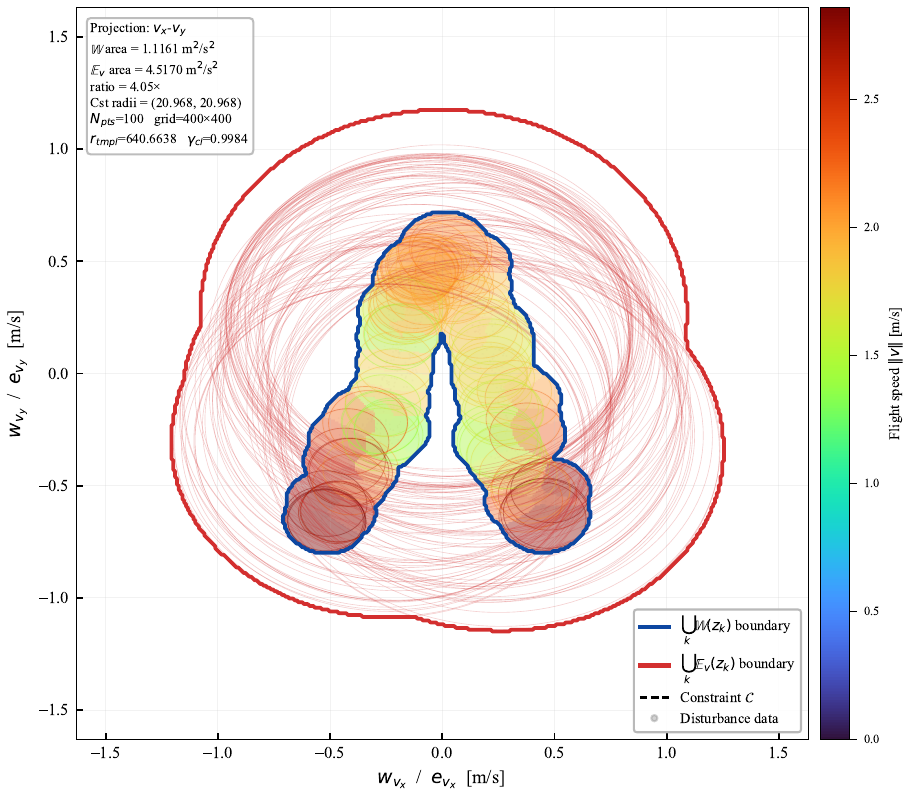}}\hfil
\subfloat[$v_x$--$v_z$]{\includegraphics[width=0.31\textwidth]{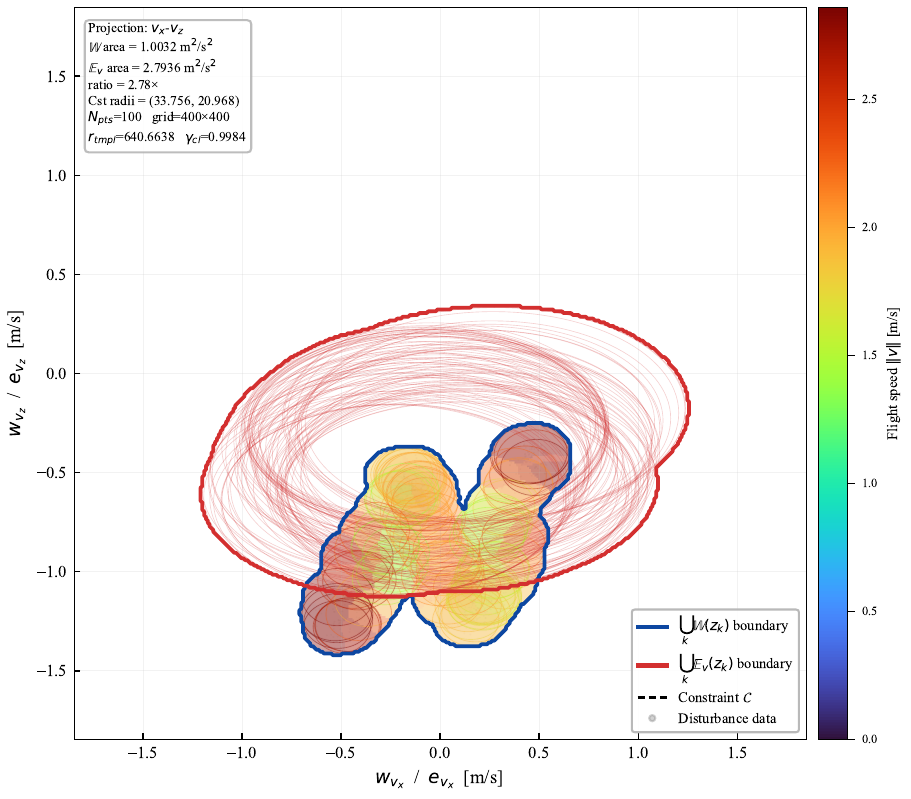}}\hfil
\subfloat[$v_y$--$v_z$]{\includegraphics[width=0.31\textwidth]{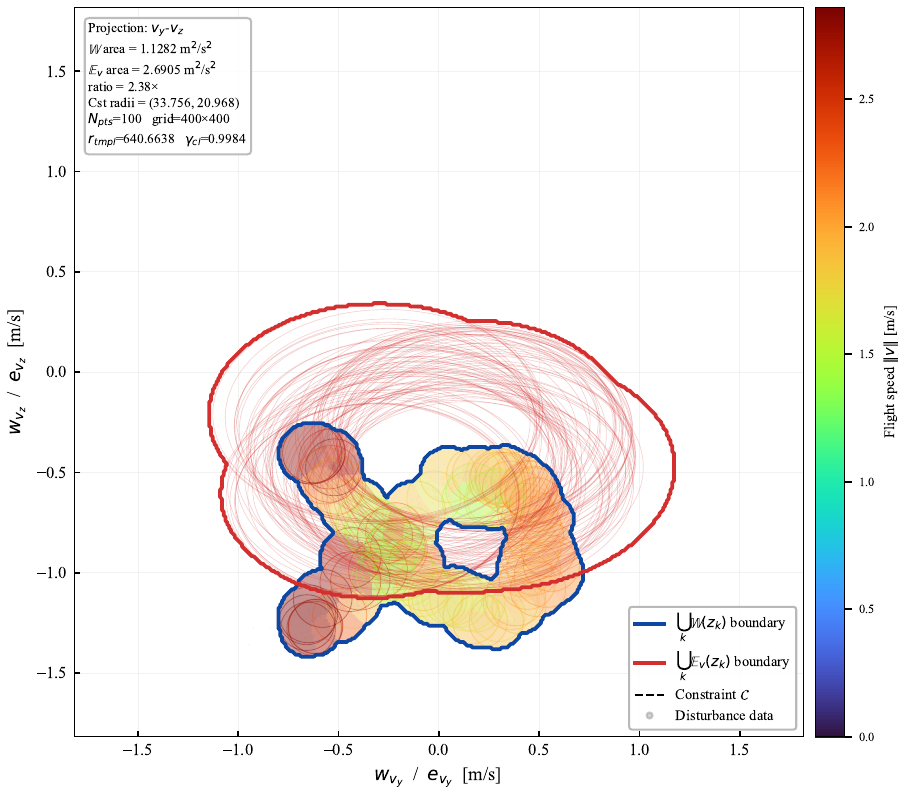}}\\
\subfloat[$\delta T$--$\delta\tau_\phi$]{\includegraphics[width=0.31\textwidth]{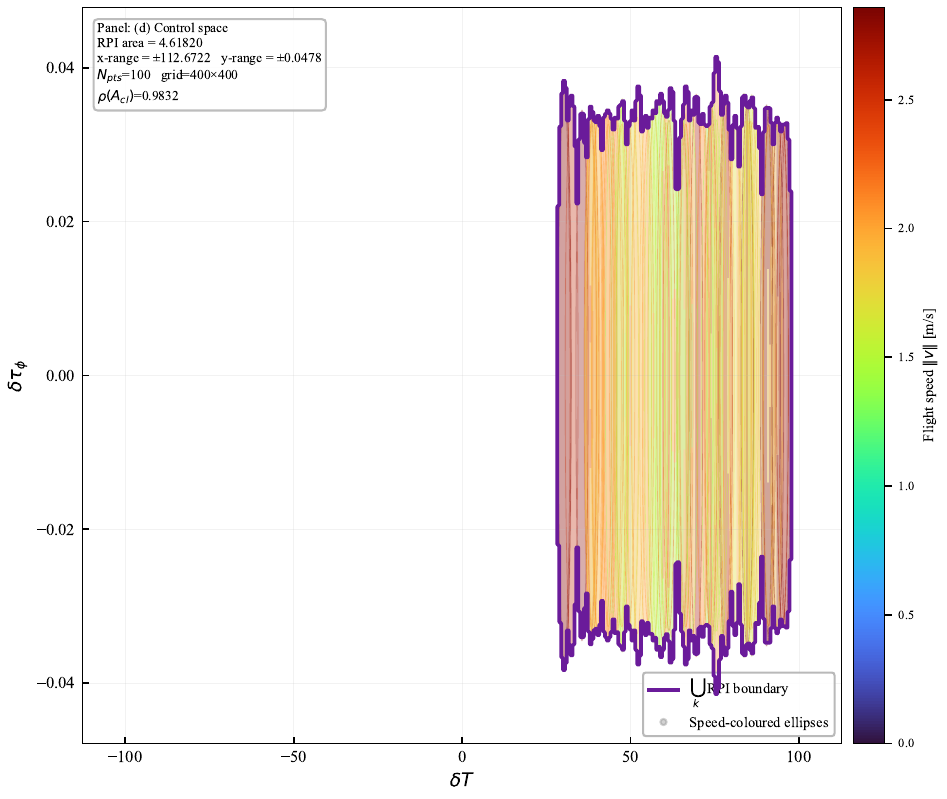}}\hfil
\subfloat[$e_{v_x}$--$\delta T$]{\includegraphics[width=0.31\textwidth]{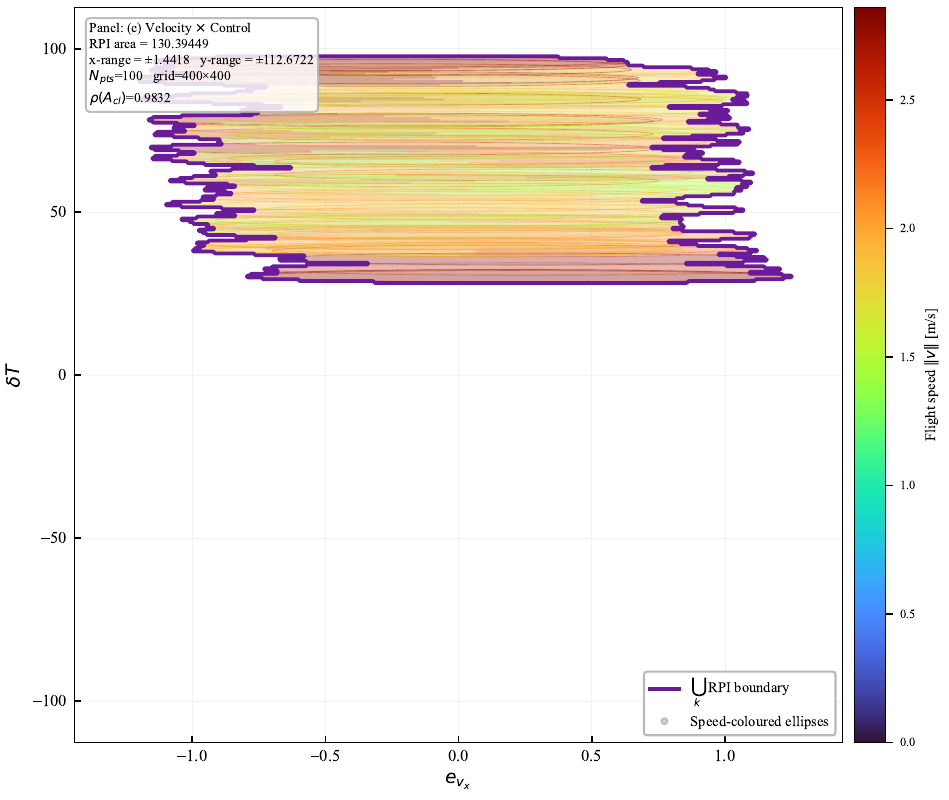}}\hfil
\subfloat[$e_{v_y}$--$\delta\tau_\phi$]{\includegraphics[width=0.31\textwidth]{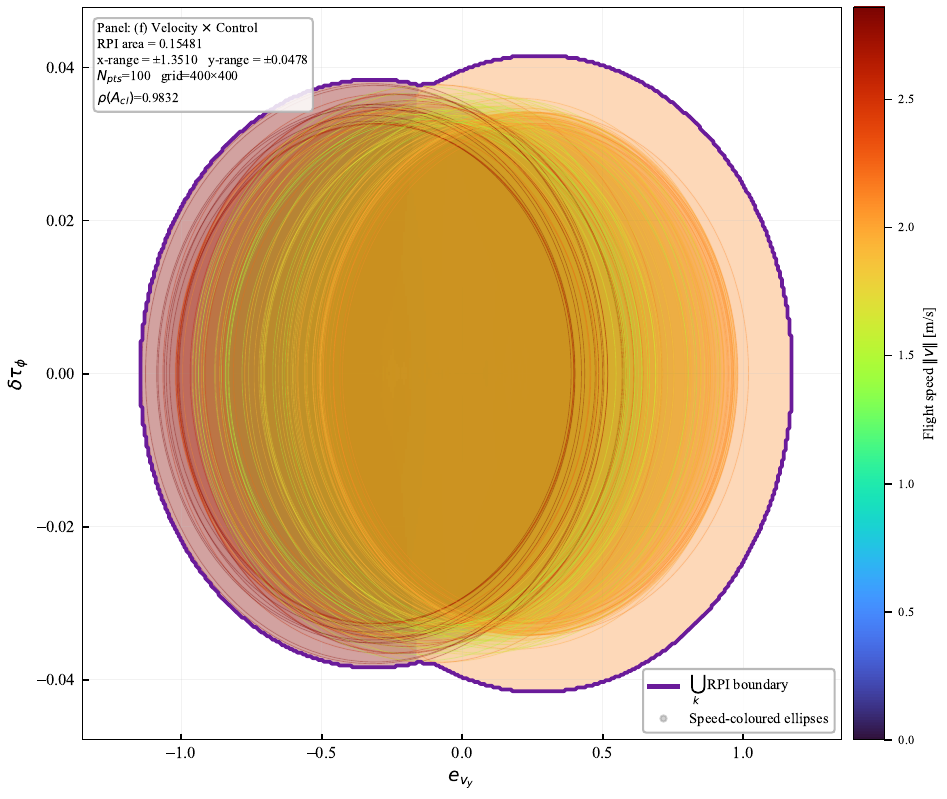}}
\caption{Per-operating-point disturbance ellipses colored by flight speed with the RPI union boundary, in velocity planes (top) and control/cross-space projections (bottom). The anisotropic template captures the coupling between velocity errors and control authority.}
\label{fig:proj}
\end{figure*}

\begin{figure*}[!t]
\centering
\subfloat[disturbance union]{\includegraphics[trim={0 0 0 1cm}, clip=true,width=0.45\textwidth]{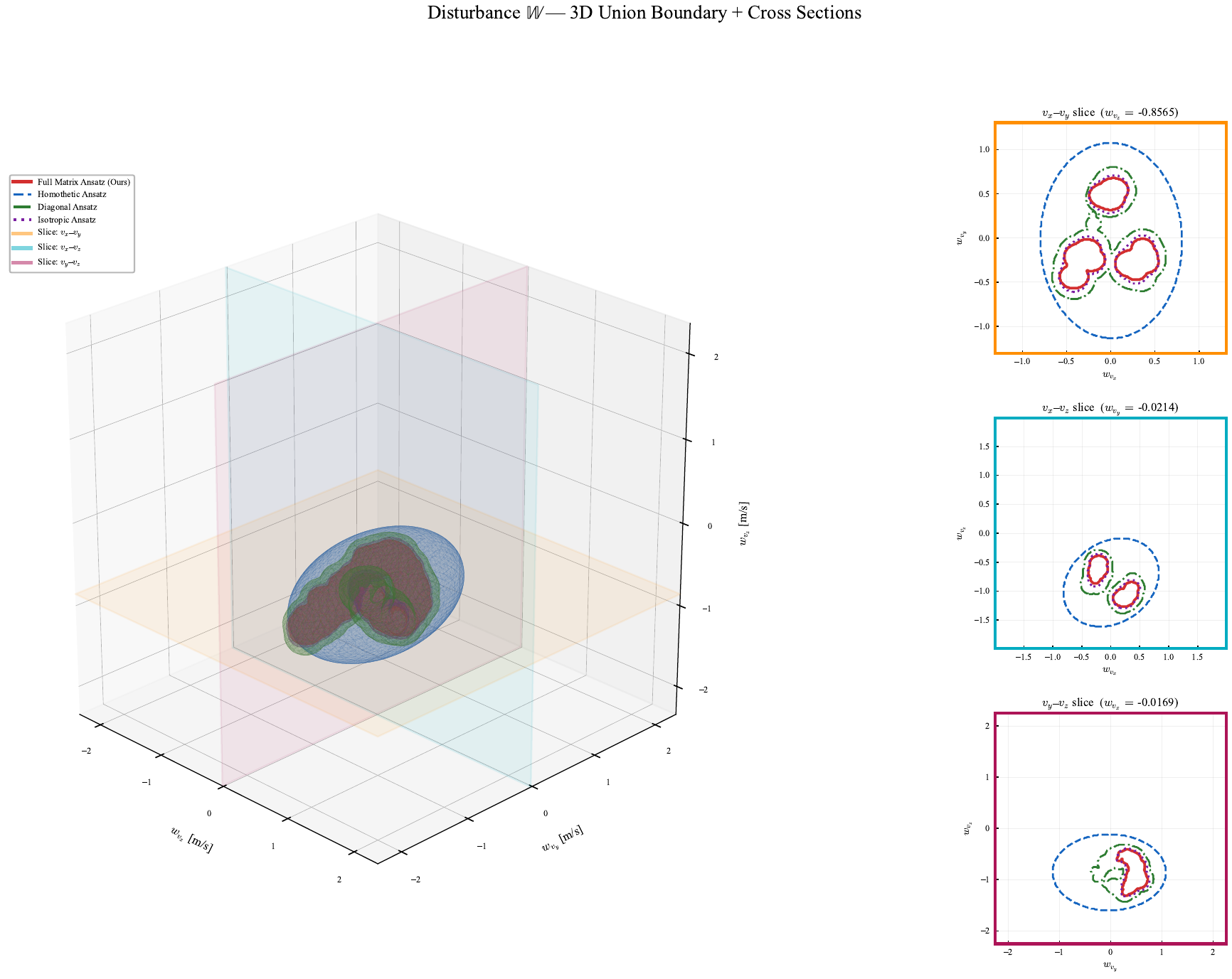}}\hfil
\subfloat[RPI error union]{\includegraphics[trim={0 0 0 1cm}, clip=true,width=0.45\textwidth]{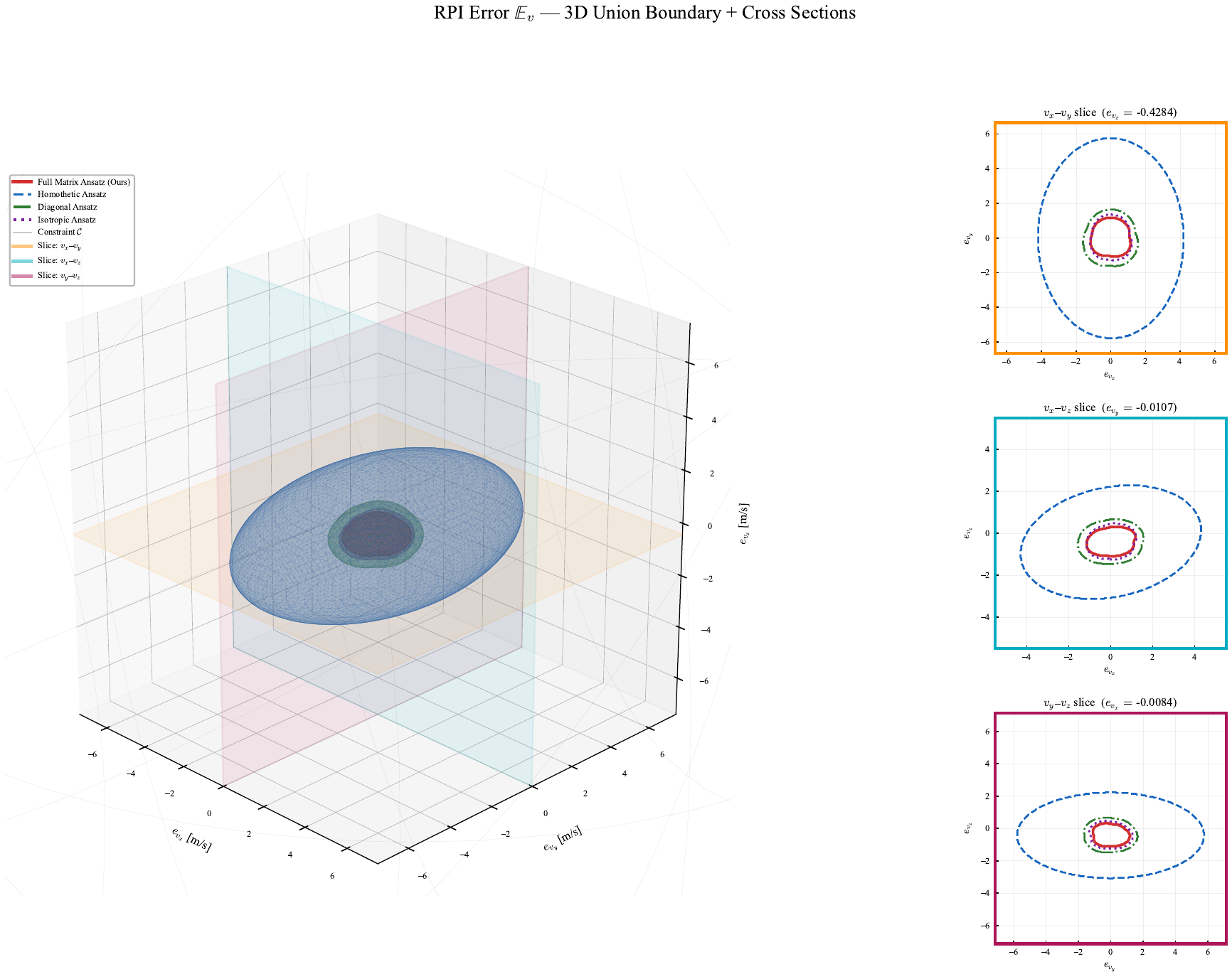}}
\caption{3D $(v_x, v_y, v_z)$ union boundary meshes via marching cubes, with cross-section insets. The anisotropic set hugs the true region while baselines over-approximate.}
\label{fig:threed}
\end{figure*}

\section{Conclusion}
\label{sec:conclusion}

The anisotropic template ansatz replaces one global tube shape with a GP-derived matrix field $\mathbf{P}(\mathbf{z})$. The graph verification keeps the certificate finite, while the online step stays local: one covariance query and one small square root. In the quadrotor study, the full-matrix field gives the largest gains precisely where scalar homothetic scaling is weakest, namely in coupled velocity-control projections. The remaining limitation is not online cost but certificate construction: the grid/graph step still needs disturbance-active coordinates or adaptive refinement in higher dimensions.

\bibliographystyle{IEEEtran}
\bibliography{ref}

\appendices

\section{Proof of Lemma~\ref{lem:contraction} (Strict Contraction)}
\label{app:contraction}

\begin{proof}
\emph{Step 1 (RPI implies Pontryagin containment).} By the defining property of the Pontryagin difference, $A \oplus B \subseteq C$ if and only if $A \subseteq C \ominus B$. Applying this to the RPI condition $\mathbf{A}_{\mathrm{cl}}\bar{\boldsymbol{\Omega}} \oplus \bar{\mathbf{W}} \subseteq \bar{\boldsymbol{\Omega}}$ with $A = \mathbf{A}_{\mathrm{cl}}\bar{\boldsymbol{\Omega}}$, $B = \bar{\mathbf{W}}$, $C = \bar{\boldsymbol{\Omega}}$ gives
\begin{equation}
\mathbf{A}_{\mathrm{cl}}\bar{\boldsymbol{\Omega}} \subseteq \bar{\boldsymbol{\Omega}} \ominus \bar{\mathbf{W}}.
\label{eq:pontryagin_step}
\end{equation}

\emph{Step 2 (Erosion by a set with interior is strict).} Since $\mathbf{0} \in \mathrm{int}(\bar{\mathbf{W}})$, there exists $\varepsilon > 0$ with $\varepsilon\mathcal{B}_2^n \subseteq \bar{\mathbf{W}}$. Let $R := \sup_{\mathbf{e} \in \bar{\boldsymbol{\Omega}}} \|\mathbf{e}\|_2 < \infty$ be the circumradius of the compact set $\bar{\boldsymbol{\Omega}}$. We claim
\begin{equation}
\bar{\boldsymbol{\Omega}} \ominus \bar{\mathbf{W}} \subseteq (1-\delta)\bar{\boldsymbol{\Omega}}, \qquad \delta := \varepsilon/R > 0.
\label{eq:erosion_strict}
\end{equation}
Take any $\mathbf{e} \in \bar{\boldsymbol{\Omega}} \ominus \bar{\mathbf{W}}$, so $\mathbf{e} + \bar{\mathbf{W}} \subseteq \bar{\boldsymbol{\Omega}}$ and in particular $\mathbf{e} + \varepsilon\mathcal{B}_2^n \subseteq \bar{\boldsymbol{\Omega}}$. Set $g := \gamma_{\bar{\boldsymbol{\Omega}}}(\mathbf{e})$. If $g = 0$ then $\mathbf{e} \in (1-\delta)\bar{\boldsymbol{\Omega}}$ trivially (note $\delta \le 1$ since $\varepsilon\mathcal{B}_2^n \subseteq \bar{\boldsymbol{\Omega}}$ forces $\varepsilon \le R$). If $g > 0$, compactness of $\bar{\boldsymbol{\Omega}}$ gives $\mathbf{e} \in g\bar{\boldsymbol{\Omega}}$, so $\mathbf{u} := \mathbf{e}/g \in \bar{\boldsymbol{\Omega}}$ with $\gamma_{\bar{\boldsymbol{\Omega}}}(\mathbf{u}) = \gamma_{\bar{\boldsymbol{\Omega}}}(\mathbf{e})/g = 1$ by positive homogeneity, and $\|\mathbf{u}\|_2 \le R$. The point $\mathbf{e} + \varepsilon\,\mathbf{u}/\|\mathbf{u}\|_2$ lies in $\mathbf{e} + \varepsilon\mathcal{B}_2^n \subseteq \bar{\boldsymbol{\Omega}}$, so its gauge is at most one. Since this point equals $(g + \varepsilon/\|\mathbf{u}\|_2)\,\mathbf{u}$, homogeneity and $\gamma_{\bar{\boldsymbol{\Omega}}}(\mathbf{u})=1$ give
\begin{equation}
\begin{aligned}
    &g + \frac{\varepsilon}{\|\mathbf{u}\|_2} \;=\; \gamma_{\bar{\boldsymbol{\Omega}}}\!\Bigl(\mathbf{e} + \varepsilon\,\tfrac{\mathbf{u}}{\|\mathbf{u}\|_2}\Bigr) \;\le\; 1
\;\; \\ &\Longrightarrow\;\;
g \le 1 - \frac{\varepsilon}{\|\mathbf{u}\|_2} \le 1 - \frac{\varepsilon}{R},
\end{aligned}
\end{equation}
which is~\eqref{eq:erosion_strict}. This makes explicit the constant in~\cite[Prop.~3.12]{blanchini2008set}.

\emph{Step 3 (Chain and conclude).} Combining~\eqref{eq:pontryagin_step} and~\eqref{eq:erosion_strict}, $\mathbf{A}_{\mathrm{cl}}\bar{\boldsymbol{\Omega}} \subseteq (1-\delta)\bar{\boldsymbol{\Omega}}$. Hence for every $\mathbf{e} \in \bar{\boldsymbol{\Omega}}$, $\gamma_{\bar{\boldsymbol{\Omega}}}(\mathbf{A}_{\mathrm{cl}}\mathbf{e}) \le 1-\delta$, and taking the supremum, $\rho_{\bar{\boldsymbol{\Omega}}} \le 1-\delta < 1$.

For the template pair of Problem~\ref{prob:main} the constant is explicit without Step~2: $\bar{\boldsymbol{\Omega}} \ominus \bar{\mathbf{W}} = r\mathcal{B}_2^n \ominus \mathcal{B}_2^n = (r-1)\mathcal{B}_2^n = (1-\tfrac{1}{r})\bar{\boldsymbol{\Omega}}$, so $\delta = 1/r$ and $\rho_{\bar{\boldsymbol{\Omega}}} \le 1 - 1/r = \gamma_{\mathrm{cl}}$.
\end{proof}

\section{Proof of Lemma~\ref{lem:ellipsoid_decomp} (Credible Ellipsoid)}
\label{app:ellipsoid}

\begin{proof}
\emph{Step 1 (Square root).} Since $\sigma_n^2 > 0$, the posterior covariance satisfies $\hat{\boldsymbol{\Sigma}}_w \in \mathbb{S}^n_{++}$ and admits a unique symmetric positive-definite square root $\hat{\boldsymbol{\Sigma}}_w^{1/2}$ with $(\hat{\boldsymbol{\Sigma}}_w^{1/2})^2 = \hat{\boldsymbol{\Sigma}}_w$~\cite[Thm.~7.2.6]{hornMatrixAnalysis2013}; its inverse $\hat{\boldsymbol{\Sigma}}_w^{-1/2}$ is also symmetric positive definite.

\emph{Step 2 (Whitening).} Define $\boldsymbol{\eta} := \hat{\boldsymbol{\Sigma}}_w^{-1/2}(\mathbf{w} - \hat{\boldsymbol{\mu}}_w)$, an invertible affine change of variables. Substituting $\mathbf{w} - \hat{\boldsymbol{\mu}}_w = \hat{\boldsymbol{\Sigma}}_w^{1/2}\boldsymbol{\eta}$,
\begin{equation}
(\mathbf{w} - \hat{\boldsymbol{\mu}}_w)^\top \hat{\boldsymbol{\Sigma}}_w^{-1}(\mathbf{w} - \hat{\boldsymbol{\mu}}_w)
= \boldsymbol{\eta}^\top \hat{\boldsymbol{\Sigma}}_w^{1/2}\hat{\boldsymbol{\Sigma}}_w^{-1}\hat{\boldsymbol{\Sigma}}_w^{1/2}\boldsymbol{\eta}
= \|\boldsymbol{\eta}\|_2^2,
\end{equation}
using $\hat{\boldsymbol{\Sigma}}_w^{1/2}\hat{\boldsymbol{\Sigma}}_w^{-1}\hat{\boldsymbol{\Sigma}}_w^{1/2} = \hat{\boldsymbol{\Sigma}}_w^{1/2}(\hat{\boldsymbol{\Sigma}}_w^{1/2})^{-1}(\hat{\boldsymbol{\Sigma}}_w^{1/2})^{-1}\hat{\boldsymbol{\Sigma}}_w^{1/2} = \mathbf{I}_n$.

\emph{Step 3 (Constraint in whitened coordinates).} The credible-region constraint becomes $\|\boldsymbol{\eta}\|_2^2 \le \chi^2_{n,1-\alpha}$, i.e.\ $\boldsymbol{\eta} \in c_{n,\alpha}\mathcal{B}_2^n$ with $c_{n,\alpha} = \sqrt{\chi^2_{n,1-\alpha}}$.

\emph{Step 4 (Invert).} Solving back, $\mathbf{w} = \hat{\boldsymbol{\mu}}_w + \hat{\boldsymbol{\Sigma}}_w^{1/2}\boldsymbol{\eta}$. As $\boldsymbol{\eta}$ ranges over $c_{n,\alpha}\mathcal{B}_2^n$, the image $\hat{\boldsymbol{\Sigma}}_w^{1/2}\boldsymbol{\eta}$ ranges over $c_{n,\alpha}\hat{\boldsymbol{\Sigma}}_w^{1/2}\mathcal{B}_2^n$ by linearity, and the translation by $\hat{\boldsymbol{\mu}}_w$ gives $\mathcal{E}(\mathbf{z}) = \hat{\boldsymbol{\mu}}_w \oplus c_{n,\alpha}\hat{\boldsymbol{\Sigma}}_w^{1/2}\mathcal{B}_2^n$.
\end{proof}

\section{Proof of Lemma~\ref{lem:P_bounded} (Boundedness)}
\label{app:bounded}

\begin{proof}
\emph{Part 1 (Upper bound).} The posterior covariance is the prior minus the explained variance,
\begin{equation}
\hat{\boldsymbol{\Sigma}}_w(\mathbf{z}) = \boldsymbol{\Sigma}_0(\mathbf{z}) - \boldsymbol{\Sigma}_{*\mathcal{D}}\bigl[\boldsymbol{\Sigma}_{\mathcal{D}\mathcal{D}} + \sigma_n^2\mathbf{I}\bigr]^{-1}\boldsymbol{\Sigma}_{\mathcal{D}*},
\label{eq:posterior_general}
\end{equation}
and the subtracted term is of the form $\mathbf{C}\mathbf{M}^{-1}\mathbf{C}^\top$ with $\mathbf{M} \succ \mathbf{0}$, hence positive semidefinite. Therefore $\hat{\boldsymbol{\Sigma}}_w(\mathbf{z}) \preceq \boldsymbol{\Sigma}_0(\mathbf{z}) \preceq \bar{k}_0\mathbf{I}_n$ by the uniform prior bound of Assumption~\ref{asm:kernel}. Since $\|\mathbf{P}\|_2 = c_{n,\alpha}\sqrt{\lambda_{\max}(\hat{\boldsymbol{\Sigma}}_w)}$, this gives $\mathbf{P}(\mathbf{z}) \preceq c_{n,\alpha}\sqrt{\bar{k}_0}\,\mathbf{I}_n =: p_{\max}\mathbf{I}_n$.

\emph{Part 2 (Lower bound: worst-case colocated dictionary).} We bound $\lambda_{\min}(\hat{\boldsymbol{\Sigma}}_w(\mathbf{z}_*))$ from below uniformly over training configurations.

\emph{Step L1.} Fix $\mathbf{z}_* \in \mathcal{Z}$ and consider the maximally informative configuration in which all $N_{\mathrm{gp}}$ active dictionary points coincide with $\mathbf{z}_*$. Writing $\boldsymbol{\Sigma}_{**} := K_{\mathrm{LMC}}(\mathbf{z}_*,\mathbf{z}_*)$,
\begin{equation}
\boldsymbol{\Sigma}_{*\mathcal{D}} = \mathbf{1}_{N_{\mathrm{gp}}}^\top \otimes \boldsymbol{\Sigma}_{**}, \qquad
\boldsymbol{\Sigma}_{\mathcal{D}\mathcal{D}} = \mathbf{1}_{N_{\mathrm{gp}}}\mathbf{1}_{N_{\mathrm{gp}}}^\top \otimes \boldsymbol{\Sigma}_{**}.
\end{equation}

\emph{Step L2 (Joint diagonalization).} Let $\boldsymbol{\Sigma}_{**} = \mathbf{V}\boldsymbol{\Lambda}\mathbf{V}^\top$ with $\boldsymbol{\Lambda} = \mathrm{diag}(\lambda_1,\ldots,\lambda_n)$. The averaged dictionary directions $\tfrac{1}{\sqrt{N_{\mathrm{gp}}}}\mathbf{1}\otimes\mathbf{v}_i$ are eigenvectors of $\boldsymbol{\Sigma}_{\mathcal{D}\mathcal{D}} + \sigma_n^2\mathbf{I}$ with eigenvalues $\sigma_n^2 + N_{\mathrm{gp}}\lambda_i$; the remaining directions lie in the nullspace of $\mathbf{1}\mathbf{1}^\top$, have eigenvalue $\sigma_n^2$, and are orthogonal to the cross-covariance with $\mathbf{z}_*$.

\emph{Step L3 (Posterior eigenvalue per direction).} Along eigendirection $i$, the cross-covariance vector is $\lambda_i(\mathbf{1}\otimes\mathbf{v}_i)$ and the explained variance evaluates to
\begin{equation}
\lambda_i^2\,(\mathbf{1}\otimes\mathbf{v}_i)^\top\bigl[\boldsymbol{\Sigma}_{\mathcal{D}\mathcal{D}}+\sigma_n^2\mathbf{I}\bigr]^{-1}(\mathbf{1}\otimes\mathbf{v}_i)
= \frac{N_{\mathrm{gp}}\lambda_i^2}{\sigma_n^2 + N_{\mathrm{gp}}\lambda_i},
\end{equation}
since $(\mathbf{1}\otimes\mathbf{v}_i)$ is an eigenvector of the bracketed matrix with eigenvalue $\sigma_n^2 + N_{\mathrm{gp}}\lambda_i$ and $\|\mathbf{1}\otimes\mathbf{v}_i\|_2^2 = N_{\mathrm{gp}}$. The posterior eigenvalue is therefore
\begin{equation}
\hat{\lambda}_i = \lambda_i - \frac{N_{\mathrm{gp}}\lambda_i^2}{\sigma_n^2 + N_{\mathrm{gp}}\lambda_i}
= \frac{\sigma_n^2\lambda_i}{\sigma_n^2 + N_{\mathrm{gp}}\lambda_i}.
\label{eq:posterior_eig}
\end{equation}
The multivariate problem decomposes into $n$ scalar computations, one per eigenvalue of $\boldsymbol{\Sigma}_{**}$.

\emph{Step L4 (Monotonicity in the prior eigenvalue).} The map $f(\lambda) := \sigma_n^2\lambda/(\sigma_n^2 + N_{\mathrm{gp}}\lambda)$ has $f'(\lambda) = \sigma_n^4/(\sigma_n^2+N_{\mathrm{gp}}\lambda)^2 > 0$, so it is strictly increasing, and the smallest posterior eigenvalue comes from the smallest prior eigenvalue:
\begin{equation}
\lambda_{\min}\bigl(\hat{\boldsymbol{\Sigma}}_w(\mathbf{z}_*)\bigr) \ge f\bigl(\lambda_{\min}(\boldsymbol{\Sigma}_{**})\bigr)
\ge \frac{\sigma_n^2\,\underline{\sigma}_0^2}{\sigma_n^2 + N_{\mathrm{gp}}\underline{\sigma}_0^2},
\end{equation}
using $\boldsymbol{\Sigma}_{**} \succeq \underline{\sigma}_0^2\mathbf{I}_n$.

\emph{Step L5 (Uniform bound).} Since $\bar{k}_0 \ge \underline{\sigma}_0^2$, replacing $\underline{\sigma}_0^2$ by $\bar{k}_0$ in the denominator only decreases the fraction, yielding the configuration-independent floor $\underline{\lambda} = \sigma_n^2\underline{\sigma}_0^2/(\sigma_n^2 + N_{\mathrm{gp}}\bar{k}_0)$ of~\eqref{eq:lambda_floor}.

\emph{Step L6 (Uniformization over active dictionaries).} 
The preceding colocated calculation gives a conservative reference bound for the most informative local configuration under the kernel assumptions used here: every active point is treated as maximally correlated with the query point, and all relevant covariance eigenvalues in the denominator are upper-bounded by $\bar{k}_0$. This produces a dictionary-independent lower envelope depending only on $(\sigma_n^2,\underline{\sigma}_0^2,\bar{k}_0,N_{\mathrm{gp}})$ rather than on the geometry of the active set. The explained variance is $\boldsymbol{\Sigma}_{*\mathcal{D}}(\boldsymbol{\Sigma}_{\mathcal{D}\mathcal{D}}+\sigma_n^2\mathbf{I})^{-1}\boldsymbol{\Sigma}_{\mathcal{D}*}$, and each cross-covariance block satisfies $\|K_{\mathrm{LMC}}(\mathbf{z}_*,\mathbf{z}^{(j)})\|_2 \le \bar{k}_0$ by positive-semidefiniteness of the joint prior (a matrix Cauchy--Schwarz bound). Colocation at $\mathbf{z}_*$ attains this ceiling for every block simultaneously, so no admissible dictionary of size at most $N_{\mathrm{gp}}$ explains more variance in the L\"owner sense than the colocated configuration. Hence the same floor
\[
\lambda_{\min}(\hat{\boldsymbol{\Sigma}}_w(\mathbf{z})) \ge \underline{\lambda}>0
\]
holds uniformly for every active dictionary with at most $N_{\mathrm{gp}}$ points and every $\mathbf{z}\in\mathcal{Z}$, and
\begin{equation}
\begin{aligned}
    \lambda_{\min}(\mathbf{P}(\mathbf{z})) &= c_{n,\alpha}\sqrt{\lambda_{\min}(\hat{\boldsymbol{\Sigma}}_w(\mathbf{z}))}\\ & \ge c_{n,\alpha}\sqrt{\underline{\lambda}} =: p_{\min} > 0.
\end{aligned}
\end{equation}
The bound needs finite $N_{\mathrm{gp}}$, $\sigma_n^2 > 0$, and a nondegenerate prior; it does not need dense data coverage of $\mathcal{Z}$.
\end{proof}

\section{Proof of Lemma~\ref{lem:P_lipschitz} (Lipschitz Continuity)}
\label{app:lipschitz}

\begin{proof}
\emph{Step 1 (Lipschitz posterior covariance).} In~\eqref{eq:posterior_general}, the matrix $\mathbf{M} := \boldsymbol{\Sigma}_{\mathcal{D}\mathcal{D}} + \sigma_n^2\mathbf{I}$ is constant in $\mathbf{z}$, while $\boldsymbol{\Sigma}_0(\cdot)$ and the cross-covariance $\mathbf{C}(\mathbf{z}) := \boldsymbol{\Sigma}_{*\mathcal{D}}(\mathbf{z})$ inherit Lipschitz continuity from the kernel. For a single output, each entry of $\mathbf{k}_*(\mathbf{z})$ is $L_k$-Lipschitz, so $\|\mathbf{k}_*(\mathbf{z}_1)-\mathbf{k}_*(\mathbf{z}_2)\|_2 \le \sqrt{N}L_k\|\mathbf{z}_1-\mathbf{z}_2\|$ and $\|\mathbf{k}_*(\mathbf{z})\|_2 \le \sqrt{N}\bar{k}$. With $\|\mathbf{M}^{-1}\|_2 \le 1/\sigma_n^2$, the explained-variance term $q(\mathbf{z}) := \mathbf{k}_*^\top\mathbf{M}^{-1}\mathbf{k}_*$ satisfies
\begin{equation}
\begin{aligned}
|q(\mathbf{z}_1)-q(\mathbf{z}_2)|
&\le \|\mathbf{M}^{-1}\|_2\bigl(\|\mathbf{k}_*(\mathbf{z}_1)\|_2 + \|\mathbf{k}_*(\mathbf{z}_2)\|_2\bigr)\nonumber\\
&\quad\times\|\mathbf{k}_*(\mathbf{z}_1)-\mathbf{k}_*(\mathbf{z}_2)\|_2
\\ &\le \frac{2N\bar{k}L_k}{\sigma_n^2}\|\mathbf{z}_1-\mathbf{z}_2\|,
\end{aligned}
\end{equation}
so the scalar posterior variance is Lipschitz with constant $L_k(1 + 2N\bar{k}/\sigma_n^2)$; the LMC blocks obey the analogous bound entrywise, giving a finite constant $L_\Sigma$ depending only on $L_k$, $N_{\mathrm{gp}}$, $\bar{k}_0$, and $\sigma_n^2$.

\emph{Step 2 (Lipschitz square root on a spectral floor).} Let $\mathbf{A},\mathbf{B} \in \mathbb{S}^n_{++}$ with $\lambda_{\min}(\mathbf{A}),\lambda_{\min}(\mathbf{B}) \ge \underline{\lambda}$, and set $\mathbf{X}=\mathbf{A}^{1/2}$, $\mathbf{Y}=\mathbf{B}^{1/2}$, $\mathbf{D}=\mathbf{X}-\mathbf{Y}$. Then
\begin{equation}
\mathbf{A}-\mathbf{B} = \mathbf{X}^2-\mathbf{Y}^2 = \mathbf{X}\mathbf{D} + \mathbf{D}\mathbf{Y},
\end{equation}
a Sylvester equation in $\mathbf{D}$. Writing it in the eigenbases $\mathbf{X}=\mathbf{U}\,\mathrm{diag}(x_i)\,\mathbf{U}^\top$ and $\mathbf{Y}=\mathbf{V}\,\mathrm{diag}(y_j)\,\mathbf{V}^\top$, the transformed entries satisfy $\tilde{D}_{ij} = \tilde{C}_{ij}/(x_i+y_j)$ where $\tilde{\mathbf{C}} = \mathbf{U}^\top(\mathbf{A}-\mathbf{B})\mathbf{V}$. Since $x_i, y_j \ge \sqrt{\underline{\lambda}}$ and the Frobenius norm is unitarily invariant,
\begin{equation}
\|\mathbf{A}^{1/2}-\mathbf{B}^{1/2}\|_F \le \frac{\|\mathbf{A}-\mathbf{B}\|_F}{2\sqrt{\underline{\lambda}}},
\end{equation}
the bound also obtainable from the integral representation of the square root~\cite{bhatiaMatrixAnalysis1997}.

\emph{Step 3 (Verify the floor and compose).} Lemma~\ref{lem:P_bounded} gives $\lambda_{\min}(\hat{\boldsymbol{\Sigma}}_w(\mathbf{z})) \ge \underline{\lambda} > 0$ on all of $\mathcal{Z}$, so Step~2 applies along any pair $\mathbf{z}_1,\mathbf{z}_2$. Composing the three Lipschitz maps $\mathbf{z} \mapsto \hat{\boldsymbol{\Sigma}}_w(\mathbf{z}) \mapsto \hat{\boldsymbol{\Sigma}}_w(\mathbf{z})^{1/2} \mapsto c_{n,\alpha}\hat{\boldsymbol{\Sigma}}_w(\mathbf{z})^{1/2}$,
\begin{equation}
\|\mathbf{P}(\mathbf{z}_1)-\mathbf{P}(\mathbf{z}_2)\|_F \le \underbrace{\frac{c_{n,\alpha}L_\Sigma}{2\sqrt{\underline{\lambda}}}}_{=:L_P}\,\|\mathbf{z}_1-\mathbf{z}_2\|.
\end{equation}
\end{proof}

\section{Proof of Lemma~\ref{lem:gauge_props} (Gauge Properties)}
\label{app:gauge}

\begin{proof}
\emph{(i)} If $\gamma_{\mathbf{S}}(\mathbf{w}) \le 1$ then, by compactness of $\mathbf{S}$, the infimum is attained and $\mathbf{w} \in \lambda\mathbf{S} \subseteq \mathbf{S}$ for some $\lambda \le 1$ (convexity with $\mathbf{0}\in\mathbf{S}$ gives $\lambda\mathbf{S} \subseteq \mathbf{S}$). Conversely $\mathbf{w}\in\mathbf{S}$ gives $\gamma_{\mathbf{S}}(\mathbf{w}) \le 1$ with $\lambda=1$.
\emph{(ii)} For $\alpha > 0$, $\mathbf{w} \in \lambda\mathbf{S} \Leftrightarrow \alpha\mathbf{w} \in \alpha\lambda\mathbf{S}$, so the feasible $\lambda$ sets scale by $\alpha$ and so do their infima; $\alpha=0$ is immediate since $\mathbf{0}\in\mathbf{S}$.
\emph{(iii)} Take $\lambda_i > \gamma_{\mathbf{S}}(\mathbf{w}_i)$ with $\mathbf{w}_i \in \lambda_i\mathbf{S}$. By convexity,
\begin{equation}
\frac{\mathbf{w}_1+\mathbf{w}_2}{\lambda_1+\lambda_2} = \frac{\lambda_1}{\lambda_1+\lambda_2}\frac{\mathbf{w}_1}{\lambda_1} + \frac{\lambda_2}{\lambda_1+\lambda_2}\frac{\mathbf{w}_2}{\lambda_2} \in \mathbf{S},
\end{equation}
so $\gamma_{\mathbf{S}}(\mathbf{w}_1+\mathbf{w}_2) \le \lambda_1+\lambda_2$; let $\lambda_i \downarrow \gamma_{\mathbf{S}}(\mathbf{w}_i)$.
\emph{(iv)} $\mathbf{w} \in \lambda\,\mathbf{M}\mathbf{S} \Leftrightarrow \mathbf{M}^{-1}\mathbf{w} \in \lambda\mathbf{S}$, so the infima coincide.
\end{proof}

\section{Proof of Theorem~\ref{thm:aniso_rpi} (Anisotropic RPI)}
\label{app:aniso_rpi}

\begin{proof}
\emph{Step 1 (Error in template coordinates).} Suppose $\mathbf{e}_k \in \boldsymbol{\Omega}(\mathbf{z}_k) = \mathbf{P}_k\bar{\boldsymbol{\Omega}}$ with $\mathbf{P}_k := \mathbf{P}(\mathbf{z}_k)$. Then there exists $\boldsymbol{\xi} \in \bar{\boldsymbol{\Omega}}$ with $\mathbf{e}_k = \mathbf{P}_k\boldsymbol{\xi}$.

\emph{Step 2 (Disturbance in template coordinates).} Likewise $\mathbf{w}_k \in \mathbb{W}(\mathbf{z}_k) = \mathbf{P}_k\bar{\mathbf{W}}$ gives $\boldsymbol{\eta} \in \bar{\mathbf{W}}$ with $\mathbf{w}_k = \mathbf{P}_k\boldsymbol{\eta}$.

\emph{Step 3 (Propagate).} The error dynamics~\eqref{eq:error_dynamics} yield $\mathbf{e}_{k+1} = \mathbf{A}_{\mathrm{cl}}\mathbf{P}_k\boldsymbol{\xi} + \mathbf{P}_k\boldsymbol{\eta}$.

\emph{Step 4 (Successor coordinates).} With $\mathbf{P}_{k+1} := \mathbf{P}(\mathbf{z}_{k+1})$ invertible, define $\tilde{\mathbf{e}}_{k+1} := \mathbf{P}_{k+1}^{-1}\mathbf{e}_{k+1}$. Then $\mathbf{e}_{k+1} \in \mathbf{P}_{k+1}\bar{\boldsymbol{\Omega}}$ iff $\tilde{\mathbf{e}}_{k+1} \in \bar{\boldsymbol{\Omega}}$.

\emph{Step 5 (Minkowski structure).} By linearity,
\begin{equation}
\tilde{\mathbf{e}}_{k+1} = \underbrace{\mathbf{P}_{k+1}^{-1}\mathbf{A}_{\mathrm{cl}}\mathbf{P}_k}_{\mathbf{M}}\,\boldsymbol{\xi} + \mathbf{P}_{k+1}^{-1}\mathbf{P}_k\,\boldsymbol{\eta},
\end{equation}
and as $(\boldsymbol{\xi},\boldsymbol{\eta})$ range over $\bar{\boldsymbol{\Omega}}\times\bar{\mathbf{W}}$, the set of attainable $\tilde{\mathbf{e}}_{k+1}$ is exactly $\mathbf{P}_{k+1}^{-1}\mathbf{A}_{\mathrm{cl}}\mathbf{P}_k\bar{\boldsymbol{\Omega}} \oplus \mathbf{P}_{k+1}^{-1}\mathbf{P}_k\bar{\mathbf{W}}$.

\emph{Step 6 (Apply the hypothesis).} Condition~\eqref{eq:aniso_rpi_cond} states this Minkowski sum lies in $\bar{\boldsymbol{\Omega}}$. Hence $\tilde{\mathbf{e}}_{k+1} \in \bar{\boldsymbol{\Omega}}$ and $\mathbf{e}_{k+1} = \mathbf{P}_{k+1}\tilde{\mathbf{e}}_{k+1} \in \boldsymbol{\Omega}(\mathbf{z}_{k+1})$.

\emph{Step 7 (Conclude).} The implication holds for every $\boldsymbol{\xi}\in\bar{\boldsymbol{\Omega}}$, $\boldsymbol{\eta}\in\bar{\mathbf{W}}$, and every transition $\mathbf{z}_k \to \mathbf{z}_{k+1}$, which is the definition of a local RPI field.
\end{proof}

\section{Proof of Lemma~\ref{lem:eigenvalue_bounds} (Eigenvalue Bounds)}
\label{app:eigenvalue_bounds}

\begin{proof}
\emph{Step 1 (Matrix perturbation bound).} For $\mathbf{z} \in \mathcal{R}_i$ with representative $\mathbf{z}_i \in \mathcal{R}_i$, Lipschitz continuity (Assumption~\ref{asm:P_regularity}, in the Frobenius norm by Lemma~\ref{lem:P_lipschitz}) and $\|\cdot\|_2 \le \|\cdot\|_F$ give
\begin{equation}
\|\mathbf{P}(\mathbf{z}) - \mathbf{P}_i\|_2 \le \|\mathbf{P}(\mathbf{z}) - \mathbf{P}_i\|_F \le L_P\|\mathbf{z}-\mathbf{z}_i\|_2 \le L_P\,\mathrm{diam}(\mathcal{R}_i) = \delta_i.
\end{equation}

\emph{Step 2 (Symmetric perturbation).} Write $\mathbf{P}(\mathbf{z}) = \mathbf{P}_i + \mathbf{E}$ with $\mathbf{E} \in \mathbb{S}^n$, $\|\mathbf{E}\|_2 \le \delta_i$.

\emph{Step 3 (Weyl).} Weyl's inequality~\cite[Cor.~4.3.15]{hornMatrixAnalysis2013} states $|\lambda_k(\mathbf{P}_i + \mathbf{E}) - \lambda_k(\mathbf{P}_i)| \le \|\mathbf{E}\|_2$ for every ordered eigenvalue index $k$. Applying it at $k$ corresponding to $\lambda_{\max}$ and $\lambda_{\min}$ gives~\eqref{eq:lambda_max_bound} and~\eqref{eq:lambda_min_bound}. Under Assumption~\ref{asm:grid_fine}, $p_i^- = \lambda_{\min}(\mathbf{P}_i) - \delta_i > 0$, so $\mathbf{P}(\mathbf{z}) \succ \mathbf{0}$ persists throughout $\mathcal{R}_i$.
\end{proof}

\section{Proof of Lemma~\ref{lem:spectral_product_bounds} (Product Bounds)}
\label{app:spectral_product}

\begin{proof}
\emph{Step 1.} For $\mathbf{z} \in \mathcal{R}_i$, $\|\mathbf{P}(\mathbf{z})\|_2 = \lambda_{\max}(\mathbf{P}(\mathbf{z})) \le p_i^+$ by Lemma~\ref{lem:eigenvalue_bounds}.

\emph{Step 2.} For $\mathbf{z}' \in \mathcal{R}_j$, positive definiteness gives $\|\mathbf{P}(\mathbf{z}')^{-1}\|_2 = 1/\lambda_{\min}(\mathbf{P}(\mathbf{z}')) \le 1/p_j^-$, well defined since $p_j^- > 0$ by Assumption~\ref{asm:grid_fine}.

\emph{Step 3.} In the working coordinates (Euclidean if $\|\mathbf{A}_{\mathrm{cl}}\|_2 < 1$, otherwise $\mathbf{P}_{\mathrm{d}}$-weighted), the closed-loop map has norm at most $\gamma_{\mathrm{cl}}$. Submultiplicativity of the spectral norm then yields
\begin{equation}
\|\mathbf{M}\|_2 \le \|\mathbf{P}(\mathbf{z}')^{-1}\|_2\,\gamma_{\mathrm{cl}}\,\|\mathbf{P}(\mathbf{z})\|_2 \le \frac{p_i^+}{p_j^-}\gamma_{\mathrm{cl}}.
\end{equation}

\emph{Step 4.} Identically, $\|\mathbf{N}\|_2 = \tfrac{1}{r}\|\mathbf{P}(\mathbf{z}')^{-1}\mathbf{P}(\mathbf{z})\|_2 \le \tfrac{1}{r}\,p_i^+/p_j^-$. Both bounds use only node-computable quantities; no evaluation of $\mathbf{P}(\cdot)$ at interior points of the regions is needed.
\end{proof}

\section{Proof of Theorem~\ref{thm:graph_invariance} (Graph Invariance)}
\label{app:graph_invariance}

\begin{proof}
\emph{Step 1 (Reduce to the norm condition).} With $\bar{\boldsymbol{\Omega}} = r\mathcal{B}_2^n$ and $\bar{\mathbf{W}} = \mathcal{B}_2^n$, condition~\eqref{eq:aniso_rpi_cond} for a transition $\mathbf{z} \in \mathcal{R}_i \to \mathbf{z}' \in \mathcal{R}_j$ reads: for all $\|\boldsymbol{\xi}\|_2 \le r$, $\|\boldsymbol{\eta}\|_2 \le 1$, $\|\mathbf{P}'^{-1}\mathbf{A}_{\mathrm{cl}}\mathbf{P}\boldsymbol{\xi} + \mathbf{P}'^{-1}\mathbf{P}\boldsymbol{\eta}\|_2 \le r$. Substituting $\boldsymbol{\xi} = r\boldsymbol{\xi}'$ with $\|\boldsymbol{\xi}'\|_2 \le 1$ and dividing by $r$,
\begin{equation}
\sup_{\|\boldsymbol{\xi}'\|_2 \le 1,\ \|\boldsymbol{\eta}\|_2 \le 1}\|\mathbf{M}\boldsymbol{\xi}' + \mathbf{N}\boldsymbol{\eta}\|_2 \le 1,
\label{eq:sup_condition_app}
\end{equation}
with $\mathbf{M},\mathbf{N}$ as in Lemma~\ref{lem:spectral_product_bounds}.

\emph{Step 2 (Triangle inequality).} For any admissible $\boldsymbol{\xi}',\boldsymbol{\eta}$,
\begin{equation}
\begin{aligned}
\|\mathbf{M}\boldsymbol{\xi}' + \mathbf{N}\boldsymbol{\eta}\|_2 \le& \|\mathbf{M}\|_2\|\boldsymbol{\xi}'\|_2 \\+ &\|\mathbf{N}\|_2\|\boldsymbol{\eta}\|_2 \le \|\mathbf{M}\|_2 + \|\mathbf{N}\|_2.
\end{aligned}
\end{equation}

\emph{Step 3 (Insert spectral bounds).} Lemma~\ref{lem:spectral_product_bounds} gives $\|\mathbf{M}\|_2 + \|\mathbf{N}\|_2 \le \tfrac{p_i^+}{p_j^-}(\gamma_{\mathrm{cl}} + \tfrac{1}{r})$, valid for \emph{all} $\mathbf{z} \in \mathcal{R}_i$, $\mathbf{z}' \in \mathcal{R}_j$, not only the nodes.

\emph{Step 4 (Arc condition).} The hypothesis~\eqref{eq:spectral_invariance_condition} bounds the right side by $1$ on every arc $(i,j) \in \mathcal{A}$, so~\eqref{eq:sup_condition_app} holds for every transition covered by an arc.

\emph{Step 5 (Completeness).} Assumption~\ref{asm:graph_complete} guarantees every dynamically feasible transition is covered by some arc. Hence~\eqref{eq:sup_condition_app} holds for all transitions in $\mathcal{Z}$, and Theorem~\ref{thm:aniso_rpi} yields local RPI of $\boldsymbol{\Omega}(\mathbf{z}) = \mathbf{P}(\mathbf{z})\,r\mathcal{B}_2^n$. The uncountable verification problem reduces to $|\mathcal{A}|$ scalar inequalities.
\end{proof}

\section{Scalar Equivalence in Theorem~\ref{thm:lmi_robust}}
\label{app:scalar_equiv}

We verify the claim that, for the scalar-bound LMI~\eqref{eq:scalar_lmi}, feasibility with $\lambda_1,\lambda_2 \ge 0$, $\lambda_1+\lambda_2 \le 1$ is equivalent to $\bar{M}_{ij} + \bar{N}_{ij} \le 1$. Write $a := \bar{M}_{ij}$, $b := \bar{N}_{ij}$. The block matrix in~\eqref{eq:scalar_lmi} equals $\mathbf{K}\otimes\mathbf{I}_n$ with
\begin{equation}
\mathbf{K} = \begin{bmatrix}\lambda_1 - a^2 & -ab\\ -ab & \lambda_2 - b^2\end{bmatrix},
\end{equation}
and $\mathbf{K}\otimes\mathbf{I}_n \succeq \mathbf{0} \Leftrightarrow \mathbf{K} \succeq \mathbf{0}$, i.e.
\begin{equation}
\lambda_1 \ge a^2, \qquad \lambda_2 \ge b^2, \qquad (\lambda_1 - a^2)(\lambda_2 - b^2) \ge a^2b^2.
\label{eq:K_psd}
\end{equation}
Minimize $\lambda_1+\lambda_2$ subject to~\eqref{eq:K_psd}. If $ab > 0$, parameterize $\lambda_1 = a^2 + t$ with $t > 0$; the product constraint forces $\lambda_2 \ge b^2 + a^2b^2/t$, so
\begin{equation}
\lambda_1+\lambda_2 \ge a^2 + b^2 + t + \frac{a^2b^2}{t} \ge a^2 + b^2 + 2ab = (a+b)^2,
\end{equation}
with equality at $t = ab$ (AM-GM). If $ab = 0$, the cross term vanishes and the minimum is $a^2+b^2 = (a+b)^2$ directly. In both cases the minimum achievable multiplier sum is $(a+b)^2$, so a feasible pair with $\lambda_1+\lambda_2 \le 1$ exists iff $(a+b)^2 \le 1$, i.e.\ $a+b \le 1$, which is~\eqref{eq:spectral_invariance_condition}. For the conservatism remark in Section~\ref{sec:invariance}: any pair with $a^2+b^2 \le 1$ satisfies $a+b \le \sqrt{2}$ by Cauchy-Schwarz, with equality at $a=b=1/\sqrt{2}$, so the triangle-inequality certificate loses at most a factor $\sqrt{2}$ against the quadratic energy bound, the worst case occurring when $\mathbf{M}\boldsymbol{\xi}$ and $\mathbf{N}\boldsymbol{\eta}$ are collinear.

\section{Proof of Proposition~\ref{prop:field_ordering} (Field Ordering)}
\label{app:field_ordering}

\begin{proof}
\emph{Part 1 (Containment).} Since $\mathbf{P}_i = c\,\boldsymbol{\Sigma}_i^{1/2}$ and $\bar{\boldsymbol{\Omega}} = r\mathcal{B}_2^n$, we have $\boldsymbol{\Omega}_i = cr\,\boldsymbol{\Sigma}_i^{1/2}\mathcal{B}_2^n$, and the positive scalars $c, r$ factor out identically on both sides; it suffices to show $\boldsymbol{\Sigma}_1^{1/2}\mathcal{B}_2^n \subseteq \boldsymbol{\Sigma}_2^{1/2}\mathcal{B}_2^n$ pointwise in $\mathbf{z}$.

\emph{Step 1a.} Let $\mathbf{e} \in \boldsymbol{\Sigma}_1^{1/2}\mathcal{B}_2^n$, so $\mathbf{e} = \boldsymbol{\Sigma}_1^{1/2}\boldsymbol{\xi}$ for some $\|\boldsymbol{\xi}\|_2 \le 1$.

\emph{Step 1b.} Define $\boldsymbol{\zeta} := \boldsymbol{\Sigma}_2^{-1/2}\mathbf{e} = \mathbf{C}\boldsymbol{\xi}$ with $\mathbf{C} := \boldsymbol{\Sigma}_2^{-1/2}\boldsymbol{\Sigma}_1^{1/2}$; we must show $\|\boldsymbol{\zeta}\|_2 \le 1$.

\emph{Step 1c (Congruence).} Because $\boldsymbol{\Sigma}_1^{1/2}$ is symmetric,
\begin{equation}
\mathbf{C}\mathbf{C}^\top = \boldsymbol{\Sigma}_2^{-1/2}\boldsymbol{\Sigma}_1\boldsymbol{\Sigma}_2^{-1/2} \preceq \boldsymbol{\Sigma}_2^{-1/2}\boldsymbol{\Sigma}_2\boldsymbol{\Sigma}_2^{-1/2} = \mathbf{I}_n,
\end{equation}
since congruence by the invertible $\boldsymbol{\Sigma}_2^{-1/2}$ preserves the L\"owner order applied to $\boldsymbol{\Sigma}_1 \preceq \boldsymbol{\Sigma}_2$.

\emph{Step 1d (Spectral norm).} $\mathbf{C}\mathbf{C}^\top$ is symmetric PSD with eigenvalues at most one, and $\|\mathbf{C}\|_2^2 = \lambda_{\max}(\mathbf{C}\mathbf{C}^\top) \le 1$. (This is the step that fails for a generic non-symmetric $\mathbf{P}_2^{-1}\mathbf{P}_1$, whose eigenvalues do not control its singular values; the hypothesis on the $\boldsymbol{\Sigma}_i$ supplies the symmetric $\mathbf{C}\mathbf{C}^\top$ form.)

\emph{Step 1e.} $\|\boldsymbol{\zeta}\|_2 = \|\mathbf{C}\boldsymbol{\xi}\|_2 \le \|\mathbf{C}\|_2\|\boldsymbol{\xi}\|_2 \le 1$, so $\mathbf{e} = \boldsymbol{\Sigma}_2^{1/2}\boldsymbol{\zeta} \in \boldsymbol{\Sigma}_2^{1/2}\mathcal{B}_2^n$. Scaling by $cr$ gives $\boldsymbol{\Omega}_1(\mathbf{z}) \subseteq \boldsymbol{\Omega}_2(\mathbf{z})$.

\emph{Part 2 (Both RPI).} Both fields satisfy~\eqref{eq:spectral_invariance_condition} by hypothesis, so Theorem~\ref{thm:graph_invariance} certifies both; the containment from Part~1 means $\boldsymbol{\Omega}_1$ gives the tighter tubes.
\end{proof}

\section{Proof of Theorem~\ref{thm:monotone} (Monotone Contraction)}
\label{app:monotone}

Part~(i) is proved in full via explicit block inversion (Route~C); two shorter alternative arguments (Routes~A and~B) follow. Parts~(ii) and~(iii) close the proof.

\subsection*{Part (i), Route C: full block-inversion derivation}

\emph{Setup.} At epoch $q$, the dataset is $\mathcal{D}^{(q)}=\{(\mathbf{z}^{(j)},\mathbf{w}^{(j)})\}_{j=1}^{N_q}$ and the posterior covariance at a test point $\mathbf{z}_*$ is
\begin{equation}
\hat{\boldsymbol{\Sigma}}_w^{(q)}(\mathbf{z}_*)=\boldsymbol{\Sigma}_{**}-\boldsymbol{\Sigma}_{*\mathcal{D}^{(q)}}\bigl[\boldsymbol{\Sigma}_{\mathcal{D}^{(q)}\mathcal{D}^{(q)}}+\sigma_n^2\mathbf{I}\bigr]^{-1}\boldsymbol{\Sigma}_{\mathcal{D}^{(q)}*},
\end{equation}
with $\boldsymbol{\Sigma}_{**} := K_{\mathrm{LMC}}(\mathbf{z}_*,\mathbf{z}_*)$. At epoch $q{+}1$ one datum $(\mathbf{z}^{\mathrm{new}},\mathbf{w}^{\mathrm{new}})$ is added.

\emph{Step 1 (Shorthand).} Fix $\mathbf{z}_*$ and define
\begin{equation}
    \begin{aligned}
\mathbf{C}_1&:=\boldsymbol{\Sigma}_{*\mathcal{D}^{(q)}}\in\mathbb{R}^{n\times nN_q}, \\
\mathbf{M}_1&:=\boldsymbol{\Sigma}_{\mathcal{D}^{(q)}\mathcal{D}^{(q)}}+\sigma_n^2\mathbf{I}\succ\mathbf{0},\nonumber\\
\mathbf{c}_2&:=K_{\mathrm{LMC}}(\mathbf{z}_*,\mathbf{z}^{\mathrm{new}})\in\mathbb{R}^{n\times n}, \\
\mathbf{s}&:=\boldsymbol{\Sigma}_{\mathcal{D}^{(q)}\mathrm{new}}\in\mathbb{R}^{nN_q\times n},\nonumber\\
\mathbf{d}&:=K_{\mathrm{LMC}}(\mathbf{z}^{\mathrm{new}},\mathbf{z}^{\mathrm{new}})+\sigma_n^2\mathbf{I}_n\succ\mathbf{0}.
\end{aligned}
\end{equation}

\emph{Step 2 (Explained variance).} Define $\boldsymbol{\Phi}^{(q)} := \boldsymbol{\Sigma}_{**} - \hat{\boldsymbol{\Sigma}}_w^{(q)}(\mathbf{z}_*) = \mathbf{C}_1\mathbf{M}_1^{-1}\mathbf{C}_1^\top \succeq \mathbf{0}$. Since $\boldsymbol{\Sigma}_{**}$ is identical at both epochs, the goal $\hat{\boldsymbol{\Sigma}}_w^{(q+1)} \preceq \hat{\boldsymbol{\Sigma}}_w^{(q)}$ is equivalent to $\boldsymbol{\Phi}^{(q+1)} \succeq \boldsymbol{\Phi}^{(q)}$.

\emph{Step 3 (Augmented system).} At epoch $q{+}1$,
\begin{equation}
\boldsymbol{\Sigma}_{*\mathcal{D}^{(q+1)}}=\begin{bmatrix}\mathbf{C}_1&\mathbf{c}_2\end{bmatrix}, \
\mathbf{M}_2:=\begin{bmatrix}\mathbf{M}_1&\mathbf{s}\\ \mathbf{s}^\top&\mathbf{d}\end{bmatrix},
\end{equation}
and $\boldsymbol{\Phi}^{(q+1)} = [\mathbf{C}_1\ \mathbf{c}_2]\,\mathbf{M}_2^{-1}\,[\mathbf{C}_1\ \mathbf{c}_2]^\top$.

\emph{Step 4 (Block inversion).} Define the Schur complement $\boldsymbol{\Delta} := \mathbf{d} - \mathbf{s}^\top\mathbf{M}_1^{-1}\mathbf{s}$. Since $\mathbf{M}_2 \succ \mathbf{0}$ and $\mathbf{M}_1 \succ \mathbf{0}$, the Schur-complement characterization of positive definiteness gives $\boldsymbol{\Delta} \succ \mathbf{0}$; expanding its definition, $\boldsymbol{\Delta} = \hat{\boldsymbol{\Sigma}}_w^{(q)}(\mathbf{z}^{\mathrm{new}}) + \sigma_n^2\mathbf{I}_n$, the epoch-$q$ posterior at the new point plus noise. The block-inversion formula~\cite[Sec.~0.7.3]{hornMatrixAnalysis2013} yields
\begin{equation}
\mathbf{M}_2^{-1}=\begin{bmatrix}
\mathbf{M}_1^{-1}+\mathbf{M}_1^{-1}\mathbf{s}\boldsymbol{\Delta}^{-1}\mathbf{s}^\top\mathbf{M}_1^{-1} & -\mathbf{M}_1^{-1}\mathbf{s}\boldsymbol{\Delta}^{-1}\\[2pt]
-\boldsymbol{\Delta}^{-1}\mathbf{s}^\top\mathbf{M}_1^{-1} & \boldsymbol{\Delta}^{-1}
\end{bmatrix}.
\end{equation}

\emph{Step 5 (Four terms).} Carrying out the block product $\boldsymbol{\Phi}^{(q+1)} = T_1+T_2+T_3+T_4$:
\begin{equation}    
\begin{aligned}
T_1 &= \mathbf{C}_1\bigl[\mathbf{M}_1^{-1}+\mathbf{M}_1^{-1}\mathbf{s}\boldsymbol{\Delta}^{-1}\mathbf{s}^\top\mathbf{M}_1^{-1}\bigr]\mathbf{C}_1^\top \nonumber\\
    &= \boldsymbol{\Phi}^{(q)} + (\mathbf{C}_1\mathbf{M}_1^{-1}\mathbf{s})\boldsymbol{\Delta}^{-1}(\mathbf{C}_1\mathbf{M}_1^{-1}\mathbf{s})^\top,\\
T_2 &= -(\mathbf{C}_1\mathbf{M}_1^{-1}\mathbf{s})\boldsymbol{\Delta}^{-1}\mathbf{c}_2^\top, \qquad T_3 = T_2^\top,\\
T_4 &= \mathbf{c}_2\boldsymbol{\Delta}^{-1}\mathbf{c}_2^\top.
\end{aligned}
\end{equation}

\emph{Step 6 (Factor the increment).} Define the predicted cross-covariance $\mathbf{Q} := \mathbf{C}_1\mathbf{M}_1^{-1}\mathbf{s}$ (the part of the $(\mathbf{z}_*,\mathbf{z}^{\mathrm{new}})$ correlation already explained by old data) and the residual $\mathbf{R} := \mathbf{c}_2 - \mathbf{Q}$ (the new information). Collecting Step~5,
\begin{equation}
\begin{aligned}
    \boldsymbol{\Phi}^{(q+1)} &= \boldsymbol{\Phi}^{(q)} + \mathbf{Q}\boldsymbol{\Delta}^{-1}\mathbf{Q}^\top - \mathbf{Q}\boldsymbol{\Delta}^{-1}\mathbf{c}_2^\top \\ &- \mathbf{c}_2\boldsymbol{\Delta}^{-1}\mathbf{Q}^\top + \mathbf{c}_2\boldsymbol{\Delta}^{-1}\mathbf{c}_2^\top,
\end{aligned}
\end{equation}
and direct expansion of $\mathbf{R}\boldsymbol{\Delta}^{-1}\mathbf{R}^\top = (\mathbf{c}_2-\mathbf{Q})\boldsymbol{\Delta}^{-1}(\mathbf{c}_2-\mathbf{Q})^\top$ reproduces the last four terms exactly. Therefore
\begin{equation}
\boldsymbol{\Phi}^{(q+1)} - \boldsymbol{\Phi}^{(q)} = \mathbf{R}\,\boldsymbol{\Delta}^{-1}\,\mathbf{R}^\top.
\label{eq:variance_increment_app}
\end{equation}

\emph{Step 7 (Positive semidefiniteness).} For any $\mathbf{v}\in\mathbb{R}^n$, with $\mathbf{w} := \mathbf{R}^\top\mathbf{v}$, $\mathbf{v}^\top\mathbf{R}\boldsymbol{\Delta}^{-1}\mathbf{R}^\top\mathbf{v} = \mathbf{w}^\top\boldsymbol{\Delta}^{-1}\mathbf{w} = \|\boldsymbol{\Delta}^{-1/2}\mathbf{w}\|_2^2 \ge 0$ since $\boldsymbol{\Delta}^{-1} \succ \mathbf{0}$.

\emph{Step 8 (Conclude, and extend to batches).} Hence $\boldsymbol{\Phi}^{(q+1)} \succeq \boldsymbol{\Phi}^{(q)}$, i.e.\ $\hat{\boldsymbol{\Sigma}}_w^{(q+1)}(\mathbf{z}_*) \preceq \hat{\boldsymbol{\Sigma}}_w^{(q)}(\mathbf{z}_*)$ for every $\mathbf{z}_* \in \mathcal{Z}$. A nested update $\mathcal{D}^{(q)} \subseteq \mathcal{D}^{(q+1)}$ adding $m \ge 1$ points is a composition of $m$ single additions (equivalently, the identical derivation with $\mathbf{c}_2, \mathbf{s}, \mathbf{d}$ of block width $nm$), so the ordering holds for arbitrary nested dictionaries.

The increment~\eqref{eq:variance_increment_app} has a direct reading: $\mathbf{R}$ is the epoch-$q$ posterior cross-covariance between $\mathbf{z}_*$ and the new point, and $\boldsymbol{\Delta}$ is the epoch-$q$ posterior variance at the new point plus noise. Information gained at $\mathbf{z}_*$ grows with residual correlation and shrinks with residual uncertainty at the new point; a fully redundant point ($\mathbf{R}=\mathbf{0}$) contributes nothing. For $n=1$ the formula collapses to the scalar update $\Phi^{(q+1)}-\Phi^{(q)} = r^2/\delta \ge 0$, the standard GP posterior-variance recursion.

\subsection*{Part (i), Route A: nested Schur complements (brief)}

The posterior $\hat{\boldsymbol{\Sigma}}_w^{(q)}(\mathbf{z}_*)$ is the Schur complement of the data block in the joint covariance
\begin{equation}
\mathbf{G}^{(q)} = \begin{bmatrix}\boldsymbol{\Sigma}_{**} & \boldsymbol{\Sigma}_{*\mathcal{D}}\\ \boldsymbol{\Sigma}_{\mathcal{D}*} & \boldsymbol{\Sigma}_{\mathcal{D}\mathcal{D}}+\sigma_n^2\mathbf{I}\end{bmatrix} \succ \mathbf{0},
\end{equation}
and $\hat{\boldsymbol{\Sigma}}_w^{(q+1)}(\mathbf{z}_*)$ is the Schur complement of the enlarged data block in the bordered matrix $\mathbf{G}^{(q+1)}$ that appends the row/column of $(\mathbf{z}^{\mathrm{new}},\mathbf{w}^{\mathrm{new}})$. Conditioning a Gaussian vector on a superset of variables cannot increase its conditional covariance: the Schur complement with respect to a larger principal block is L\"owner-dominated by the one with respect to the smaller block. This nested-Schur-complement monotonicity is precisely the identity computed in Route~C, whose increment $\mathbf{R}\boldsymbol{\Delta}^{-1}\mathbf{R}^\top \succeq \mathbf{0}$ quantifies the gap; Route~A is its coordinate-free statement.

\subsection*{Part (i), Route B: information form (brief)}

Let $\boldsymbol{\Lambda}^{(q)}(\mathbf{z}_*) := [\hat{\boldsymbol{\Sigma}}_w^{(q)}(\mathbf{z}_*)]^{-1}$. Incorporating one observation updates the information matrix as $\boldsymbol{\Lambda}^{(q+1)} = \boldsymbol{\Lambda}^{(q)} + \mathbf{H}^\top\boldsymbol{\Gamma}^{-1}\mathbf{H}$, the multi-output information-filter recursion, where $\mathbf{H}$ encodes the GP cross-covariance with the new point and $\boldsymbol{\Gamma} \succ \mathbf{0}$ is the innovation covariance; the added term is PSD since $\mathbf{v}^\top\mathbf{H}^\top\boldsymbol{\Gamma}^{-1}\mathbf{H}\mathbf{v} = \|\boldsymbol{\Gamma}^{-1/2}\mathbf{H}\mathbf{v}\|_2^2 \ge 0$. Hence $\boldsymbol{\Lambda}^{(q+1)} \succeq \boldsymbol{\Lambda}^{(q)}$. Inversion reverses the L\"owner order on $\mathbb{S}^n_{++}$: if $\mathbf{A} \succeq \mathbf{B} \succ \mathbf{0}$ then $\mathbf{B}^{-1/2}\mathbf{A}\mathbf{B}^{-1/2} \succeq \mathbf{I}$, so its inverse $\mathbf{B}^{1/2}\mathbf{A}^{-1}\mathbf{B}^{1/2} \preceq \mathbf{I}$, i.e.\ $\mathbf{A}^{-1} \preceq \mathbf{B}^{-1}$. Applying this to the information growth gives $\hat{\boldsymbol{\Sigma}}_w^{(q+1)} \preceq \hat{\boldsymbol{\Sigma}}_w^{(q)}$.

\subsection*{Part (ii): anisotropy field shrinkage}

\emph{Step 2a.} $\mathbf{P}^{(q)}(\mathbf{z}) = c_{n,\alpha}[\hat{\boldsymbol{\Sigma}}_w^{(q)}(\mathbf{z})]^{1/2}$ with scalar $c_{n,\alpha} > 0$.

\emph{Step 2b.} The square root $t \mapsto t^{1/2}$ is operator monotone on $\mathbb{S}^n_+$ by the L\"owner-Heinz theorem~\cite[Sec.~5.3]{bhatiaPositiveDefiniteMatrices2015}: $\mathbf{A} \preceq \mathbf{B} \Rightarrow \mathbf{A}^{1/2} \preceq \mathbf{B}^{1/2}$. (Operator monotonicity fails for exponents above one; $\mathbf{A} \preceq \mathbf{B}$ does not imply $\mathbf{A}^2 \preceq \mathbf{B}^2$, which is why this step is where the structure of~\eqref{eq:P_def} is used.)

\emph{Steps 2c-2d.} Applying Part~(i) and then operator monotonicity, $[\hat{\boldsymbol{\Sigma}}_w^{(q+1)}]^{1/2} \preceq [\hat{\boldsymbol{\Sigma}}_w^{(q)}]^{1/2}$; multiplying by $c_{n,\alpha} > 0$ preserves the order, giving $\mathbf{P}^{(q+1)}(\mathbf{z}) \preceq \mathbf{P}^{(q)}(\mathbf{z})$.

\subsection*{Part (iii): tube contraction}

Part~(i) gives $\hat{\boldsymbol{\Sigma}}_w^{(q+1)}(\mathbf{z}) \preceq \hat{\boldsymbol{\Sigma}}_w^{(q)}(\mathbf{z})$ for all $\mathbf{z}$. Proposition~\ref{prop:field_ordering} with $\boldsymbol{\Sigma}_i = \hat{\boldsymbol{\Sigma}}_w^{(i)}$ and $c = c_{n,\alpha}$ yields the pointwise inclusion $\boldsymbol{\Omega}^{(q+1)}(\mathbf{z}) = \mathbf{P}^{(q+1)}(\mathbf{z})\bar{\boldsymbol{\Omega}} \subseteq \mathbf{P}^{(q)}(\mathbf{z})\bar{\boldsymbol{\Omega}} = \boldsymbol{\Omega}^{(q)}(\mathbf{z})$. Each ellipsoidal cross-section nests inside its predecessor, with directions of high data density shrinking fastest. \hfill$\blacksquare$

\section{Proof of Proposition~\ref{prop:preservation} (Preservation)}
\label{app:preservation}

\begin{proof}
\emph{Step 1.} For this proof the uniform bounds are taken over the spectral-bound parameters of Definition~\ref{def:spectral_bounds}: $p_{\max}^{(q)} := \max_{i\in\mathcal{V}} p_i^+$ and $p_{\min}^{(q)} := \min_{i\in\mathcal{V}} p_i^-$, with $p_i^+ = \lambda_{\max}(\mathbf{P}_i)+\delta_i$ and $p_j^- = \lambda_{\min}(\mathbf{P}_j)-\delta_j$. These node bounds already absorb the Lipschitz deviation, so they dominate the continuous values $\sup_\mathbf{z}\|\mathbf{P}^{(q)}(\mathbf{z})\|_2$ and $\inf_\mathbf{z}\lambda_{\min}(\mathbf{P}^{(q)}(\mathbf{z}))$.

\emph{Steps 2-3.} With the Step~1 definitions, $p_i^+ \le p_{\max}^{(q)}$ and $p_j^- \ge p_{\min}^{(q)}$ hold directly, since $p_{\max}^{(q)}$ and $p_{\min}^{(q)}$ are the extreme node bounds. Lemma~\ref{lem:eigenvalue_bounds} guarantees these node bounds enclose $\lambda_{\max}(\mathbf{P}(\mathbf{z}))$ and $\lambda_{\min}(\mathbf{P}(\mathbf{z}))$ throughout each region $\mathcal{R}_i$.

\emph{Step 4.} For any arc $(i,j) \in \mathcal{A}$,
\begin{equation}
\frac{p_i^+}{p_j^-}\Bigl(\gamma_{\mathrm{cl}}+\frac{1}{r}\Bigr) \le \frac{p_{\max}^{(q)}}{p_{\min}^{(q)}}\Bigl(\gamma_{\mathrm{cl}}+\frac{1}{r}\Bigr) \le 1,
\end{equation}
the last inequality being the hypothesis.

\emph{Step 5.} The arc-wise condition~\eqref{eq:spectral_invariance_condition} therefore holds on every arc. The uniform condition is sufficient, not necessary: individual arcs may certify even when the global ratio violates the uniform bound, which is exactly the slack the graph-based test exploits.
\end{proof}

\end{document}